\documentclass[a4paper, 12pt]{article}

\usepackage[sort&compress]{natbib}
\bibpunct{(}{)}{;}{a}{}{,} 

\usepackage{amsthm, amsmath, amssymb, mathrsfs, multirow, url, subfigure}
\usepackage{graphicx, algorithm, algorithmicx, algpseudocode} 
\usepackage{ifthen} 
\usepackage{amsfonts}
\usepackage[usenames]{color}
\usepackage{fullpage}

\numberwithin{equation}{section} 

\theoremstyle{plain}

\newtheorem{prop}{\indent Proposition}

\theoremstyle{definition}

\theoremstyle{remark}

\newtheorem*{astep}{\indent A-step}
\newtheorem*{pstep}{\indent P-step}
\newtheorem*{cstep}{\indent C-step}

\newcommand{\prob}{\mathsf{P}}

\newcommand{\bel}{\mathsf{bel}}
\newcommand{\pl}{\mathsf{pl}}

\newcommand{\pois}{{\sf Pois}}
\newcommand{\unif}{{\sf Unif}}
\newcommand{\nm}{{\sf N}}

\newcommand{\RR}{\mathbb{R}}

\newcommand{\XX}{\mathbb{X}}
\newcommand{\UU}{\mathbb{U}}

\newcommand{\TT}{\mathbb{T}}

\renewcommand{\S}{\mathcal{S}}
\renewcommand{\SS}{\mathbb{S}}

\newcommand{\del}{\partial}
\renewcommand{\phi}{\varphi} 
\newcommand{\eps}{\varepsilon}

\title{Optimal inferential models for a Poisson mean}
\author{
Ryan Martin \\
Department of Mathematics, Statistics, and Computer Science \\
University of Illinois at Chicago \\
{\tt rgmartin@math.uic.edu} \\
\mbox{} \\
Duncan Ermini Leaf \mbox{} and \mbox{} Chuanhai Liu \\
Department of Statistics \\
Purdue University \\
{\tt $\{$dleaf,\,chuanhai$\}$@stat.purdue.edu}
}
\date{\today}

\begin{document}

\maketitle 

\begin{abstract}  
Statistical inference on the mean of a Poisson distribution is a fundamentally important problem with modern applications in, e.g., particle physics.  The discreteness of the Poisson distribution makes this problem surprisingly challenging, even in the large-sample case.  Here we propose a new approach, based on the recently developed framework of \emph{inferential models} (IMs).  Specifically, we construct optimal, or at least approximately optimal, IMs for two important classes of assertions/hypotheses about the Poisson mean.  For point assertions, we develop a novel recursive sorting algorithm to construct this optimal IM.  Numerical comparisons of the proposed method to existing methods are given, for both the mean and the more challenging mean-plus-background problem.   

\smallskip

\end{abstract}

\section{Introduction}
\label{S:intro}

Statistical inference based on discrete data, in particular, Poisson counts, is a fundamentally important and counterintuitively challenging problem.   For example, modern inference problems in high-energy physics involve Poisson count data, and the combination of discreteness, small sample size, and occasional parameter constraints cause trouble for classical frequentist methods; see \citet{mandelkern2002}, \citet{bcd2003}, and the references therein.  Bayesian methods, popular in part for their conceptual and computational simplicity, also suffer in such problems because, in addition to the uncertain choice of prior, the inferential output generally is not calibrated for easy interpretation by users.  So, these kinds of challenging problems apparently require new ways of handling uncertainty.  In this paper, we apply the recently developed framework of \emph{inferential models} (IMs) to this problem of inference on a Poisson mean.  

The primary goal of statistical inference is the conversion of experience, in the form of observed data, into scientific knowledge.  But in order for a consensus to ultimately be reached, it is desirable that the inferential output, i.e., measures of uncertainty about the truthfulness of any assertion/hypothesis of interest, be meaningful both within and across experiments.  
\begin{itemize}
\item[I.] \emph{Meaningfulness within an experiment}.  The inferential output should depend on the observed data in a logical and meaningful way.  For example, Bayesian posterior probabilities or p-values can, in principle, be plotted as functions of observed data, and sense can be made out of the relationships revealed in this plot; e.g., a hypothesis is more plausible for one data value than for another.  On the other hand, frequentist hypothesis testing procedures, and the conclusions reached by them, are justified based Type~I and Type~II error rates, which are calculated pre-data and, therefore, meaningless in the given problem.   
\vspace{-2mm}
\item[II.] \emph{Meaningfulness across experiments}.  Inferential outputs should be suitably calibrated so that, if many similar experiments are conducted at different times or places, then the data-dependent measure of support for a true (resp.~false) assertion should be large (resp.~small) for a majority of the experiments, where ``large/small'' and ``majority'' have mathematical definitions available pre-experiment.  The language of frequentist error rates can be used to describe such properties, but it is not the frequentist properties themselves that are important, but rather the interpretability of the inferential results that is derived from them.  
\end{itemize}

As mentioned above, frequentist methods generally fail to satisfy Property~I.  In discrete data problems, such as Poisson, frequentist methods also tend to violate Property~II: typically large-sample approximations are used, which may not be appropriate in applications, and extreme care must be taken even if they are appropriate \citep{bcd2003}.  Bayesian methods satisfy Property~I, but without a carefully chosen reference prior, there are no guarantees that Property~II can be satisfied.  Other methods for probabilistic inference are available, namely, Fisher's fiducial inference \citep{fisher1973, zabell1992}, its variants \citep{hannig2009}, and Dempster--Shafer theory \citep{dempster2008, shafer1976}.  These methods generally produce output which is meaningful in the sense of Property~I.  However, to be meaningful, fiducial probabilities must be interpreted subjectively and, therefore, do not generally satisfy the calibration in Property~II.  

The IM framework of \citet{imbasics} was built upon ideas first laid out in \citet{mzl2010} and \citet{zl2010}.  The term ``inferential model'' reflects the understanding that an inferential method satisfying both Properties~I and II generally requires something more than fiducial's ``continue to regard'' \citep{dempster1963} strategy.  \citet{imbasics} develop a general and relatively simple three-step construction of an IM.  The details of this construction are reviewed in Section~\ref{S:review}.  As a result of this careful reasoning with uncertainty, the IM framework identifies and corrects the inherent bias in Fisher's fiducial inference.  Moreover, under very mild conditions, this IM output is shown to satisfy both desirable Properties~I and II.  

In this paper we specialize the general IM framework to the important Poisson problem, extending the naive analysis of this problem in \citet{imbasics} in two directions.  After a brief introduction to the basic IM construction and theoretical properties in Section~\ref{S:review}, we present results on optimal IM construction for two important classes of assertions/hypotheses about the Poisson mean, namely, one- and two-sided assertions.  Section~\ref{S:one.sided} establishes a simple result on the optimal IM for one-sided assertions.  The more challenging class of two-sided assertions is considered in Section~\ref{S:two.sided}.  There we develop first some intuitions about the optimal IM construction, and then propose a novel recursive algorithm for construction of an (approximately) optimal IM for two-sided assertions, which translates directly to interval estimates for the Poisson mean.  Our second contribution is an extenstion to the problem where non-stochastic constraint information about the Poisson mean is available, in addition to the observed data.  This constrained Poisson mean problem has applications in high-energy physics, where signal counts cannot be directly distinguished from background noise.  Numerical comparisons in Section~\ref{SS:background} show that the proposed method compares favorably to existing methods in terms of a variety of frequentist criteria.  However, it is important to keep in mind that IMs are more than just a tool to construct frequentist procedures: IMs produce prior-free posterior probabilistic inference, exactly what Fisher's fiducial inference was designed to achieve.

\section{Brief review of IMs}
\label{S:review}

\subsection{Definitions and basic construction}
\label{SS:basics}

Building on ideas in \citet{mzl2010} and \citet{zl2010}, \citet{imbasics} presented a general framework of prior-free, posterior probabilistic inference based on what are called inferential models (IMs).  To fix notation, let $X$ be the observable data, taking values in a space $\XX$, and let $\theta$ be the parameter of interest, taking values in the parameter space $\Theta$.  Given the application we have in mind here, we shall assume $\Theta$ and $\XX$ are subsets of $\RR$.  The starting point of the IM framework is similar to that of fiducial, in the sense that an auxiliary variable, denoted by $U$ and taking values in a space $\UU$ with probability measure $\prob_U$, is associated with $X$ and $\theta$.  It is this association, together with the distribution $U \sim \prob_U$, which characterizes the sampling distribution $X \sim \prob_{X|\theta}$.  After observing $X=x$, the fiducial/Dempster--Shafer approach is to ``continue to regard'' \citep{dempster1963} $U$ as a sample from $\prob_U$, and then invert the association to get a corresponding fiducial posterior distribution for $\theta$, given $X=x$.  

The IM approach takes a different perspective.  That is, instead of keeping the interpretation of $U$ as a random variable, the IM approach treats the unobserved value $u^\star$ of $U$, which is tied to the observed data $X=x$ and the \emph{true value} of $\theta$, as the fundamental quantity.  Then the goal is to predict this unobserved value $u^\star$ with a random set.  It turns out that the success of the IM framework rests on the choice of this predictive random set, described in more detail next.  

Start with a collection $\SS = \{S_t: t \in \TT\}$ of $\prob_U$-measurable subsets of $\UU$, indexed by some generic space $\TT$.  This collection will serve as the support of the predictive random set.  \citet{imbasics} showed that, for optimal predictive random sets, it suffices to assume that the collection $\SS$ is nested in the sense that either $S_t \subseteq S_{t'}$ or $S_{t'} \subseteq S_t$ for all pairs $t,t' \in \TT$.  We can define now define the predictive random set $\S$, supported on $\SS$, with ``distribution function'' $\prob_\S\{\S \subseteq S\} = \prob_U(S)$, for $S \in \SS$, what we call the natural measure.  Any predictive random set constructed in this way is \emph{admissible}; the name ``admissible'' is based on the result \citep[][Theorem~3]{imbasics} that for any predictive random set, there is one in this admissible class that is as good or better.  Therefore, without loss of efficiency, we may restrict attention to predictive random sets with nested supports equipped with the natural measure.  

The following three steps, described in \citet{imbasics}, define an IM:

\begin{astep}
Associate $X$, $\theta$, and $U \sim \prob_U$ in a way consistent with the sampling distribution $X \sim \prob_{X|\theta}$ such that for all $x \in \XX$ and all $u \in \UU$, it defines a unique subset $\Theta_x(u) \subseteq \Theta$, possibly empty, containing all possible candidate values of $\theta$ given $(x,u)$.  
\end{astep}

\begin{pstep}
Predict the unobserved value $u^\star$ of $U$ associated with the observed data by an admissible  predictive random set $\S$.  
\end{pstep}

\begin{cstep}
Combine $\S$ and the association $\Theta_x(u)$ specified in the A-step to obtain 
\begin{equation}
\label{eq:new.focal}
\Theta_x(\S) = \bigcup_{u \in \S} \Theta_x(u).
\end{equation}
Then compute the \emph{belief function}
\begin{equation}
\label{eq:belief}
\bel_x(A; \S) = \prob_\S\{\Theta_x(\S) \subseteq A\}, 
\end{equation}
where $A \subseteq \Theta$ is the assertion/hypothesis about $\theta$ of interest.  
\end{cstep}

%The use of the larger predictive random set $\S^+$ in \eqref{eq:new.focal}, due to \citet{leafliu2012}, provides a way of dealing with so-called ``conflict cases,'' i.e., cases where there are no $\theta$ values consistent with a particular $x$ and a draw of $\S$.  Such cases did not arise in the examples in \citet{imbasics}, but this adjustment is often needed in constrained parameter problems; see Section~\ref{SS:background}.  

The belief function is just one part of the inferential output.  Since the belief function $\bel_x(A; \S)$ is sub-additive, i.e., $\bel_x(A;\S) + \bel_x(A^c;\S) \leq 1$, one actually needs both $\bel_x(A;\S)$ and $\bel_x(A^c;\S)$ to summarize the information in $x$ concerning the truthfulness of assertion $A$.  In some cases, it is more convenient to report the \emph{plausibility function} 
\begin{equation}
\label{eq:plausibility}
\pl_x(A;\S) = \prob_\S\{\Theta_x(\S) \cap A \neq \varnothing\} = 1-\bel_x(A^c; \S).
\end{equation}
Often, Monte Carlo methods are required to evaluate the belief/plausibility functions.  Also note that it is not necessary to have the same predictive random set for each of $A$ and $A^c$.  In fact, for optimal inference, \citet{imbasics} recommend using different predictive random sets for each point in $\Theta$; see Section~\ref{S:two.sided}.

\subsection{Validity and optimality}
\label{SS:validity}

The performance of a particular predictive random set is measured through the sampling behavior of the corresponding belief function, as a function of $X \sim \prob_{X|\theta}$, at a given assertion $A$.  In particular, the IM is said to be \emph{valid at $A$} if 
\begin{equation}
\label{eq:valid}
\sup_{\theta \in A^c} \prob_{X|\theta}\{\bel_X(A;\S) \geq 1-\alpha\} \leq \alpha, \quad \alpha \in (0,1), 
\end{equation}
or, in other words, $\bel_X(A;\S)$ is stochastically no larger than $\unif(0,1)$ when $X \sim \prob_{X|\theta}$ with $\theta \not\in A$.  This validity property is a mathematical description of Property~II in Section~\ref{S:intro}.  That is, if $A$ is \emph{false}, then the amount of support in data $X$ for $A$ will be large only for a relatively small proportion of $X$ values.  \citet[][Theorem~1]{imbasics} show that this validity property is easy to arrange: it holds whenever the predictive random set $\S$ is admissible in the sense described above.  

As a consequence of the validity theorem, one can use the IM output--belief and plausibility functions---to construct frequentist decision procedures.  For example, in a testing problem, $H_0: \theta \in A$ versus $H_1: \theta \not\in A$, the testing rule 
\begin{equation}
\label{eq:imtest}
\text{reject $H_0$ based on $X=x$ iff $\pl_x(A;\S) \leq \alpha$}
\end{equation}
controls the frequentist Type~I error rate at the nominal $\alpha$ level.  One can also construct a $100(1-\alpha)$\% plausibility region for the unknown parameter by inverting this test,
\[ \Pi_x(\alpha) = \{\theta: \pl_x(\theta; \S) > \alpha\}. \]
This plausibility region also has nominal frequentist coverage probability; see \citet{imbasics} for details.  But we should emphasize here that, although plausibility functions can be used to construct frequentist procedures, they can also do much more.  Indeed, the belief and plausibility functions provide meaningful prior-free posterior probabilistic evidence for the truthfulness of the claim ``$\theta \in A$.''  In particular, any $\theta' \not\in \Pi_x(\alpha)$ is a relatively implausible value for the true $\theta$ after observing $X=x$.  Confidence/credible intervals simply do not have this sharp of an interpretation.   

Herein we focus only on IMs that are valid in the sense of \eqref{eq:valid}.  In that case, $\bel_X(A;\S)$, as a function of $X$, is (probabilistically) not too large when $A$ is false.  Towards optimality, we want $\bel_X(A;\S)$ as large as possible without violating the validity condition.  For this, a non-trivial upper bound on the belief function will be helpful.  Given $A$, define a class of subsets of $\UU$ indexed by $x \in \XX$:
\begin{equation}
\label{eq:a.event}
\UU_x(A) = \{u \in \UU: \Theta_x(u) \subseteq A\}.
\end{equation}
In words, $\UU_x(A)$ contains all those $u$ such that, given $x$, the corresponding $\theta$ values all agree with the assertion $A$.  It can be shown that $\prob_U\{\UU_x(A)\}$ is the fiducial/Dempster--Shafer posterior probability for $A$, given data $x$.  This fiducial probability can also be written as an IM belief function, i.e., 
\begin{equation}
\label{eq:fiducial}
\prob_U\{\UU_x(A)\} = \bel_x(A;\S_0), \quad \text{where $\S_0 = \{U\}$, $U \sim \prob_U$}.  
\end{equation}
\citet[][Proposition~1]{imbasics} show that, for any admissible predictive random set $\S$, $\bel_x(A;\S)$ is bounded above by $\prob_U\{\UU_x(A)\}$ for all $x$.  If it happens that $\{\UU_x(A): x \in \XX\}$ is nested, then an admissible predictive random set $\S^\star$ exists such that the upper bound is attained, i.e., $\bel_x(A;\S^\star) = \bel_x(A;\S_0)$ for all $x$.  In this case, we say that the IM corresponding to $\S^\star$ is optimal.  We summarize this result as follows.  

\begin{prop}
\label{prop:optimal}
Given an assertion $A$, suppose that $\{\UU_x(A):x \in \XX\}$ defined in \eqref{eq:a.event} forms a nested collection of sets.  Then there exists an admissible predictive random set $\S^\star$ such that $\bel_x(A;\S^\star) = \bel_x(A;\S_0)$ for all $x$.  
\end{prop}

\begin{proof}[\indent Proof]
Take the index set $\TT = \XX$ and define the support $\SS = \{\UU_x(A): x \in \XX\}$.  This collection is nested by hypothesis.  Take $\S^\star$ to be the predictive random set determined by the natural measure as in \eqref{eq:natural.measure}.  Then $\S^\star$ is admissible.  Furthermore, 
\[ \bel_x(A;\S^\star) = \prob_{\S^\star}\{\Theta_x(\S^\star) \subseteq A\} = \prob_{\S^\star}\{\S^\star \subseteq \UU_x(A)\} = \prob_U\{\UU_x(A)\}. \]
Since the right-hand side equals $\bel_x(A;\S_0)$, the claim follows.  
\end{proof}

Proposition~\ref{prop:optimal} resolves this issue of optimal IMs in problems where $\{\UU_x(A):x \in \XX\}$, is nested; see Section~\ref{S:one.sided}.  However, in other cases, like in Section~\ref{S:two.sided}, these sets are not nested so further considerations are needed.  \citet{imbasics} develop an theory of optimal IMs for continuous data models, and steps towards optimality in the discrete Poisson data problem are discussed in Section~\ref{S:two.sided}.

\section{Poisson inference for one-sided assertions}
\label{S:one.sided}

\subsection{A simple Poisson association}

For the Poisson model, $X \sim \pois(\theta)$, the probability mass function is $f_\theta(x) = e^{-\theta} \theta^x/x!$, $x=0,1,2,\ldots$, and the distribution function $F_\theta(x)$ satisfies 
\[ F_\theta(x) = 1-G_{x+1}(\theta), \quad x=0,1,2,\ldots, \quad \theta > 0, \]
where $G_a$ is the gamma distribution function with scale parameter $a$ and rate parameter unity.  Following \citet{imbasics}, we introduce $U \sim \prob_U = \unif(0,1)$, and define the association between data $X$, parameter $\theta$, and auxiliary variable $U$ as 
\begin{equation}
\label{eq:poisson.association}
F_\theta(X-1) \leq 1-U < F_\theta(X), \quad U \sim \unif(0,1).
\end{equation}
It is clear that this association characterizes the posited Poisson sampling model; this is the familiar recipe for simulation from the Poisson distribution.  Using the connection between the Poisson and gamma distribution functions, we can rewrite \eqref{eq:poisson.association}, for generic $(x,\theta,u)$, as $G_{x+1}(\theta) < u \leq G_x(\theta)$, and, by inversion, we have
\begin{equation}
\label{eq:pois.focal1}
\Theta_x(u) = \bigl[ G_x^{-1}(u), G_{x+1}^{-1}(u) \bigr),  
\end{equation}
the set of all candidate $\theta$'s, given $(x,u)$.

\subsection{Optimal IMs}

Let $\theta_0 > 0$ be an arbitrary but fixed value, and consider the assertion $A = (\theta_0, \infty)$.  This assertion is ``one-sided'' in the same sense that the alternative hypothesis $H_1: \theta > \theta_0$ in the classical testing context is one-sided.  In this case, using \eqref{eq:pois.focal1}, the sets $\UU_x(A)$ defined in \eqref{eq:a.event} are given by 
\[ \UU_x(A) = \{u: G_x^{-1}(u) > \theta_0\} = \{u: u > G_x(\theta_0)\} = \{u: u > 1 - F_{\theta_0}(x-1)\}. \]
Since $F_{\theta_0}(\cdot)$ is a non-decreasing function, it follows that $\UU_x(A) \subset \UU_{x'}(A)$ for non-negative integers $x < x'$.  Since these sets are nested, there is an optimal IM that can be obtained as in the proof of Proposition~\ref{prop:optimal}.  This optimal IM has belief function 
\[ \bel_x(A; \S_A^\star) = \prob_U\{\UU_x(A)\} = F_{\theta_0}(x-1), \quad \prob_U = \unif(0,1). \]
Here we use the notation $\S_A^\star$ to denote the predictive random set corresponding to the optimal IM for the assertion $A=(\theta_0,\infty)$.  

Now consider $A^c=(0,\theta_0]$, the alternate one-sided assertion.  Calculations similar to those displayed above shows that $\UU_x(A^c) = \{u: u \leq 1-F_{\theta_0}(x)\} = (0, 1-F_{\theta_0}(x)]$.  Since these again are nested, the optimal IM for $A^c$ has belief function
\[ \bel_x(A^c; \S_{A^c}^\star) = \prob_U\{\UU_x(A^c)\} = 1-F_{\theta_0}(x). \]

To summarize, for the one-sided assertion $A=(\theta_0,\infty)$, an optimal IM exists and can be found via Proposition~\ref{prop:optimal}.  Specifically, for a given $X=x$, the corresponding optimal belief and plausibility function pair is given by 
\[ \{\bel_x(A), \pl_x(A)\} = \{F_{\theta_0}(x-1), F_{\theta_0}(x)\}. \]
Some connections between the IM results and classical hypothesis testing are worth mentioning here.  First, observe that the plausibility function is exactly Fisher's p-value for testing the null hypothesis $H_0: \theta \in A$.  That is, the p-value can be interpreted as an upper bound on the belief probability that the null hypothesis is true.  Second, as described in \citet{imbasics}, an IM-based frequentist testing rule would reject $H_0: \theta \in A$ based on observed $X=x$ if the plausibility function $\pl_x(A)$ is too small, i.e., if $\pl_x(A) \leq \alpha$.  They show that such a testing rule controls the frequentist Type~I error at level $\alpha$.  But, in addition, if we ignore randomization issues, then this same rule with $\pl_x(A) = F_{\theta_0}(x)$, corresponds to the Neyman--Pearson most powerful test.

\section{Poisson inference for two-sided assertions}
\label{S:two.sided}

Consider a singleton assertion $A=\{\theta_0\}$ for some fixed $\theta_0 > 0$.  This corresponds to a point null hypothesis $H_0: \theta=\theta_0$ like in the classical setting.  It is well known that point nulls and, hence, singleton assertions are closely tied to the important problem of constructing confidence/plausibility intervals.  In this section we will focus our attention on the complement $A^c=\{\theta_0\}^c$, a so-called ``two-sided'' assertion.  

For this two-sided assertion, the sets $\UU_x(\{\theta_0\}^c)$ are 
\begin{align}
\UU_x(\{\theta_0\}^c) & = \{u: G_{x+1}^{-1}(u) \leq \theta_0\} \cup \{u: G_x^{-1}(u) > \theta_0\} \notag \\
& = \{u: u \leq G_{x+1}(\theta_0)\} \cup \{u: u > G_x(\theta_0)\} \notag \\
& = (0,1) \setminus (G_{x+1}(\theta_0), G_x(\theta_0)]. \label{eq:pois.aevent}
\end{align}
It is clear from the latter expression that $\UU_x(\{\theta_0\}^c)$ are not nested.  Therefore, Proposition~\ref{prop:optimal} does not help to identify an optimal IM---something more is needed.

\subsection{Nesting predictive random sets via intersections}
\label{SS:intersections}

Following the intuition developed in Proposition~\ref{prop:optimal}, we see that the use of the sets $\{\UU_x(\{\theta_0\}^c): x\in \XX\}$ is desirable.  But in order for the corresponding belief function to be valid, these sets need to be modified to make them nested.  One way this can be accomplished is by iteratively taking intersections, i.e., order the sets $\{\UU_{x_k}(\{\theta_0\}^c): k \geq 1\}$ and define $S_1 = \UU_{x_1}(\{\theta_0\}^c)$, $S_2 = \UU_{x_2}(\{\theta_0\}^c) \setminus S_1^c$, and so on.  The following two-step procedure describes this idea in more detail.   

\begin{enumerate}
\item Choose a ranking $\rho$ on $\XX$, i.e., an ordering of $\{\UU_x(\{\theta_0\}^c): x \in \XX\}$.  
\vspace{-2mm}
\item Let $\TT=\{1,2,\ldots\}$ and define $\SS_\rho = \{S_t^\rho: t \in \TT\}$ as follows.  Set $S_0 = \varnothing$ and  
\[ S_t^\rho = \bigcap_{x: \rho(x) > t} \UU_x(\{\theta_0\}^c) = \bigcup_{x: \rho(x) \leq t} (G_{x+1}(\theta_0), G_x(\theta_0)], \quad t=1,2,\ldots, \]
where the last equality follows from \eqref{eq:pois.aevent}.  
\end{enumerate}

For each $\rho$, the collection $\SS_\rho$ is nested, so if it is equipped with the natural measure \eqref{eq:natural.measure}, then we obtain an admissible predictive random set $\S_\rho$.  Since $S_{\rho(x)-1}^\rho$ is the largest of the $S_r^\rho$'s that is contained in $\UU_x(\{\theta_0\}^c)$, it follows that 
\begin{align*}
\bel_x(\{\theta_0\}^c; \S_\rho) & = \prob_U\{S_{\rho(x)-1}^\rho\} = \sum_{x': \rho(x') < \rho(x)} [G_{x'}(\theta_0)-G_{x'+1}(\theta_0) ] = \sum_{x': \rho(x') < \rho(x)} f_{\theta_0}(x'),
\end{align*}
and, consequently, the corresponding plausibility function is 
\[ \pl_x(\theta_0; \S_\rho) \equiv \pl_x(\{\theta_0\};\S_\rho) = 1-\sum_{x': \rho(x') < \rho(x)} f_{\theta_0}(x'). \]
It follows from the general theory that the IM based on $\S_\rho$ is valid for any ranking $\rho$.  Following \citet{imbasics}, the optimal $\rho$ is such that $\bel_X(\{\theta_0\}^c;\S_\rho)$ is largest (probabilistically) under $X \sim \pois(\theta)$, $\theta\neq\theta_0$.  %Here the fiducial probability---the belief function upper bound---is $\prob_U\{\UU_x(\{\theta_0\}^c)\} = 1-f_{\theta_0}(x)$.  

\subsection{Optimal ordering: some intuition}

Towards an optimal ordering, we consider the distribution of $\bel_X(\{\theta_0\}^c;\S_\rho)$ as a function of $X \sim \prob_{X|\theta} = \pois(\theta)$, for $\theta \neq \theta_0$.  Consider the event $\{\bel_X(\{\theta_0\}^c;\S_\rho) \leq \bel_x(\{\theta_0\}^c; \S_\rho)\}$, for a given $x \in \XX$.  Then the $\prob_{X|\theta}$-probability of this event is like the distribution function of $\bel_X(\{\theta_0\}^c; \S_\rho)$, i.e., 
\begin{equation}
\label{eq:dist.fun}
\psi_x(\theta) = \prob_{X|\theta}\{\bel_X(\{\theta_0\}^c; \S_\rho) \leq \bel_x(\{\theta_0\}^c; \S_\rho)\} = \sum_{x': \rho(x') < \rho(x)} f_\theta(x'), 
\end{equation}
which we treat as a function of $\theta$ for each fixed $x$; the dependence on the ranking $\rho$ will be implicit in the notation.  For optimality, we want the belief function to be as large as possible without breaking the validity requirement.  So we follow \citet{imbasics} and impose on $\rho$ the condition that 
\begin{equation}
\label{eq:condition1}
\text{$\psi_x(\theta)$ is maximized at $\theta=\theta_0$ for each $x$}. 
\end{equation}
By \eqref{eq:condition1}, the derivative of $\psi_x(\theta)$ with respect to $\theta$ vanishes at $\theta_0$, i.e., 
\begin{equation}
\label{eq:score.balance}
\sum_{x': \rho(x') < \rho(x)} T_{\theta_0}(x') f_{\theta_0}(x') = 0, \quad \forall\; x \in \XX, 
\end{equation}
where $T_\theta(x) = (\del/\del\theta) \log f_\theta(x) = x/\theta-1$ is the score function.  Recall that, in many cases, including the Poisson example considered here, the score function has zero expectation.  Therefore, we refer to \eqref{eq:score.balance} as the \emph{score-balance condition}---that is, in order to satisfy \eqref{eq:score.balance}, the ranking $\rho$ must be suitably symmetric, or balanced, with respect to the sampling distribution of $T_{\theta_0}(X)$ under $X \sim \pois(\theta_0)$.  

By \eqref{eq:condition1}, the second derivative of $\psi_x(\theta)$ with respect to $\theta$,  at $\theta=\theta_0$ satisfies
\begin{equation}
\label{eq:V.inequality}
\sum_{x': \rho(x') < \rho(x)} V_{\theta_0}(x') f_{\theta_0}(x') < 0, \quad \forall\; x \in \XX, 
\end{equation}
where $V_{\theta_0}(x) = T_{\theta_0}(x)^2 + (\del/\del\theta) T_\theta(x) \bigr|_{\theta=\theta_0}$.  Consequently, the ranking $\rho$ must be chosen so that \eqref{eq:V.inequality} holds in addition to \eqref{eq:score.balance}.  Following a remark about notation, we give some intuition for how this can be accomplished.  

In what follows, for ease of interpretation, we report the algorithm and numerical results with the current parametrization of $\theta$, the mean of the Poisson distribution.  However, it is more convenient theoretically to work with the natural parameter in the exponential family representation.  So by working first with parameter $\eta=\log\theta$, i.e., differentiating with respect to $\eta$, and then substituting $\theta=e^\eta$, we have 
\begin{equation}
\label{eq:new.par}
T_{\theta_0}(x) = x-\theta_0 \quad \text{and} \quad V_{\theta_0}(x) = (x-\theta_0)^2 - \theta_0.
\end{equation}
These expressions are different from what is obtained by working with $\theta$ throughout.  

In order to achieve \eqref{eq:V.inequality}, the basic idea is to choose $\rho$ such that $x$ values with small values of $|T_{\theta_0}(x)| = |x-\theta_0|$ are assigned higher rank.  This is based on the fact that $V_{\theta_0}(x) = (x-\theta_0)^2-\theta_0$ is a quadratic in $T_{\theta_0}(x)$, and so $V_{\theta_0}(x)$ is smallest for $x$ with small absolute score.  The problem is not this simple, unfortunately, because this intuition fails to account for the multiplication by the probability mass function in \eqref{eq:V.inequality}.  Due to the discreteness, an optimal ranking $\rho^\star$ satisfying both \eqref{eq:score.balance} and \eqref{eq:V.inequality} does not exist in general.  But the formal algorithm described in the following subsection recursively defines a permutation that \emph{approximately} achieves this optimal ordering.

\subsection{Optimal ordering: a recursive scheme}
\label{SS:recursive}

Here we construct an increasing sequence $\{E_r: r \geq 0\}$ of subsets of $\XX$, with $E_0=\varnothing$.  From these, the (approximately) optimal ranking $\rho^\star$ is obtained as $\rho^\star(E_r \setminus E_{r-1}) = r$.  

Recall that, here, we are working with the abused notation described above.  That is, we start out with the Poisson distribution indexed by the natural parameter $\eta$, the log of the mean, and then substitute $\theta=e^\eta$ back into the expressions for the score function, etc.  Define two subsets of $\XX$:
\[ \XX^+ = \{x \in \XX: T_{\theta_0}(x) \geq 0\} \quad \text{and} \quad \XX^- = \{x \in \XX: T_{\theta_0}(x) < 0\}. \]
These sets with non-negative and negative scores will be updated iteratively in the algorithm that follows.  The basic idea is to choose $E_r$, containing elements of both $\XX^+$ and $\XX^-$, in such a way that \eqref{eq:score.balance} and \eqref{eq:V.inequality} hold, at least approximately.  Algorithm~\ref{algo:sort} gives the details.  R code to implement this procedure is available at \url{www.math.uic.edu/~rgmartin}.  Line~22 stops the algorithm if both proxies---$\nu_r(1)$ and $\nu_r(2)$---for the left-hand side of \eqref{eq:V.inequality} are positive.  In our experience, no such error will occur.  

\begin{algorithm*}[t]
\smallskip
Given tolerance $\eps > 0$, take finite $\XX^\eps \subset \XX$ such that $\prob_{X|\theta_0}\{\XX^\eps\} \geq 1-\eps$.  
\begin{algorithmic}[1]
\State {\bf initialize} $\XX_0^+ = \XX^+ \cap \XX^\eps$, $\XX_0^- = \XX^- \cap \XX^\eps$, $E_0 = \varnothing$, $r=1$;
\While {$r \leq \#(\XX^\eps)$} 
\If {$\XX_{r-1}^+ = \varnothing$} 
\State $E_r = E_{r-1} \cup \{\max \XX_{r-1}^-\}$;
\State $\XX_r^- = \XX_{r-1}^- \setminus \{\max \XX_{r-1}^-\}$;
\ElsIf {$\XX_{r-1}^- = \varnothing$}
\State $E_r = E_{r-1} \cup \{\min \XX_{r-1}^+\}$;
\State $\XX_r^+ = \XX_{r-1}^+ \setminus \{\min \XX_{r-1}^+\}$;
\Else
\State $E_r(1) = E_{r-1} \cup \{\min\XX_{r-1}^+\}$;
\State $E_r(2) = E_{r-1} \cup \{\max\XX_{r-1}^-\}$;
\For {$k=1,2$} 
\State $\tau_r(k) = (\del/\del\theta) \log \sum_{x \in E_r(k)} f_\theta(x) \bigr|_{\theta=\theta_0}$;
\State $\nu_r(k) = \sum_{x \in E_r(k)} V_{\theta_0}(x) f_{\theta_0}(x)$; 
\EndFor
\If {$|\tau_r(1)| \leq |\tau_r(2)|$ {\bf and} $\nu_r(1) \leq 0$} 
\State $E_r = E_r(1)$;
\State $\XX_r^+ = \XX_{r-1}^+ \setminus \{\min \XX_{r-1}^+\}$;
\ElsIf {$\nu_r(2) \leq 0$}
\State $E_r = E_r(2)$;
\State $\XX_r^- = \XX_{r-1}^- \setminus \{\max \XX_{r-1}^-\}$;
\Else \; {\bf stop} %\emph{something is wrong!}
\EndIf
\EndIf
\State $r = r+1$;
\EndWhile
\end{algorithmic}
\caption{\bf-- Recursive ordering.}
\label{algo:sort}
\end{algorithm*}

As we described previously, it is intuitively clear that the recursive ordering scheme in Algorithm~\ref{algo:sort} will determine an ordering $\rho=\rho^\star$ such that \eqref{eq:score.balance} and \eqref{eq:V.inequality} approximately hold.  Here we do a numerical check to confirm this claim.  Let 
\begin{equation}
\label{eq:TV.r}
T(r) = \sum_{x \in E_r} T_{\theta_0}(x) f_{\theta_0}(x) \quad \text{and} \quad V(r) = \sum_{x \in E_r} V_{\theta_0}(x) f_{\theta_0}(x), 
\end{equation}
where $E_r$ is constructed as in Algorithm~\ref{algo:sort}, and $T_{\theta_0}(x)$ and $V_{\theta_0}(x)$ are as in \eqref{eq:new.par}.  If \eqref{eq:score.balance} and \eqref{eq:V.inequality} hold, then we expect $T(r)$ to be close to 0 and $V(r)$ to be negative, respectively, for all $r$.  Figure~\ref{fig:checks} plots $T(r)$ and $V(r)$ as functions of $r$, and, indeed, our expectations are mostly realized.  At first look, the fluctuations in $T(r)$ seem a bit troubling, but it turns out that these are effectively dampened by the magnitude of $V(r)$.  To see this, let $\psi_r(\theta)$ be the $\psi_x(\theta)$ in \eqref{eq:dist.fun} such that $\rho(x)=r$.  A two-term Taylor approximation of $\psi_r(\theta)$ at $\theta=\theta_0$ can be written as 
\begin{align*}
\psi_r(\theta)-\psi_r(\theta_0) %& = T(r) (\theta-\theta_0) + (1/2) V(r) (\theta-\theta_0)^2 + o(|\theta-\theta_0|^2) \\
& = V(r)(\theta-\theta_0) \Bigl[ \frac{T(r)}{V(r)} + \frac{\theta-\theta_0}{2} \Bigr] + o(|\theta-\theta_0|^2).
\end{align*}
So if $V(r) < 0$ and $T(r)/V(r)$ is close to zero, respectively, for each $r$, then the difference should be negative and, hence, $\psi_r(\theta)$ is maximized at $\theta=\theta_0$ for each $r$.  From Figure~\ref{fig:checks} it is clear that $T(r)/V(r)$ has smaller fluctuations than $T(r)$.  

\begin{figure}[t]
\begin{center}
\subfigure[$T(r)$ vs.~$r$, with $\theta_0=5$]{\scalebox{0.6}{\includegraphics{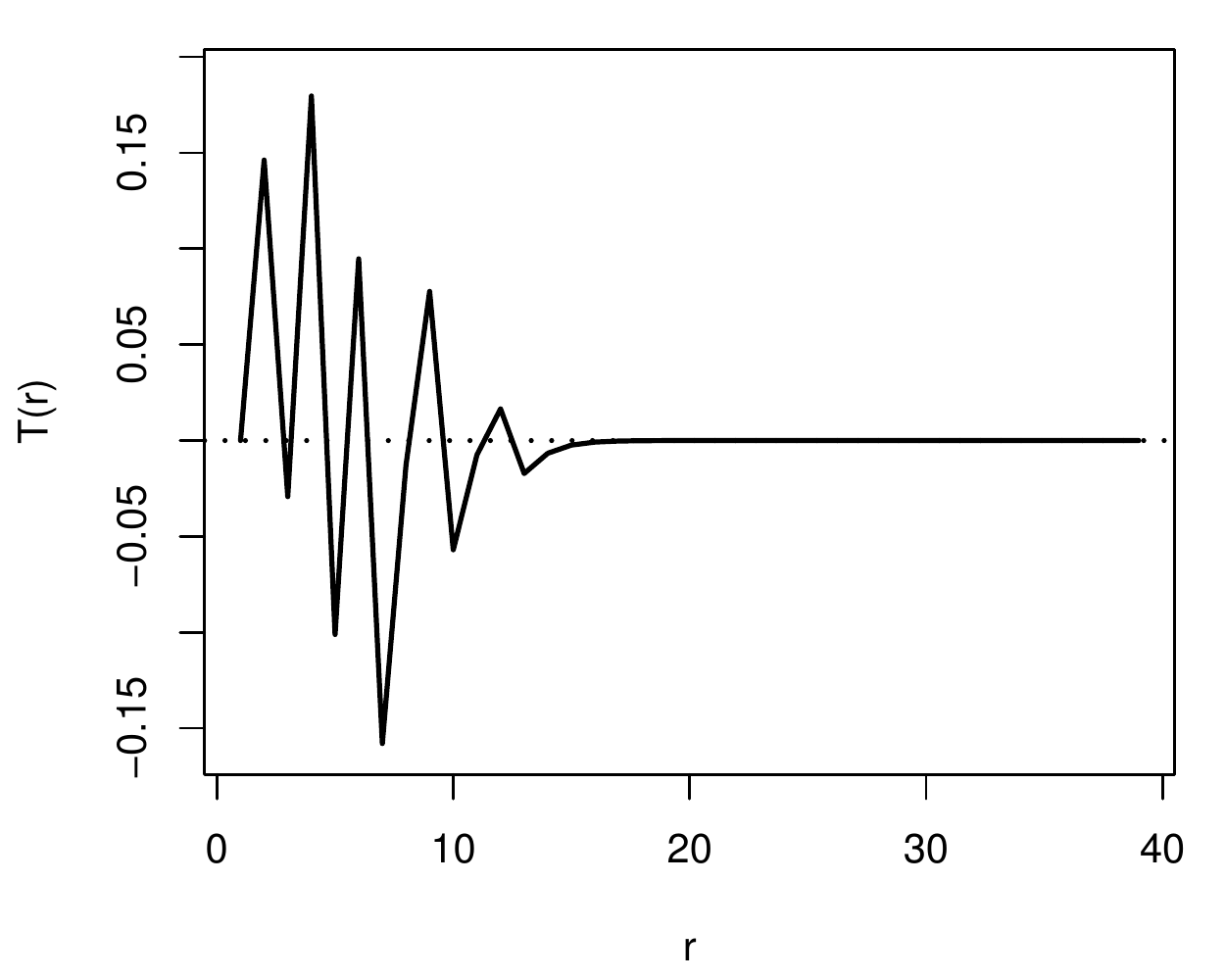}}}
\subfigure[$V(r)$ vs.~$r$, with $\theta_0=5$]{\scalebox{0.6}{\includegraphics{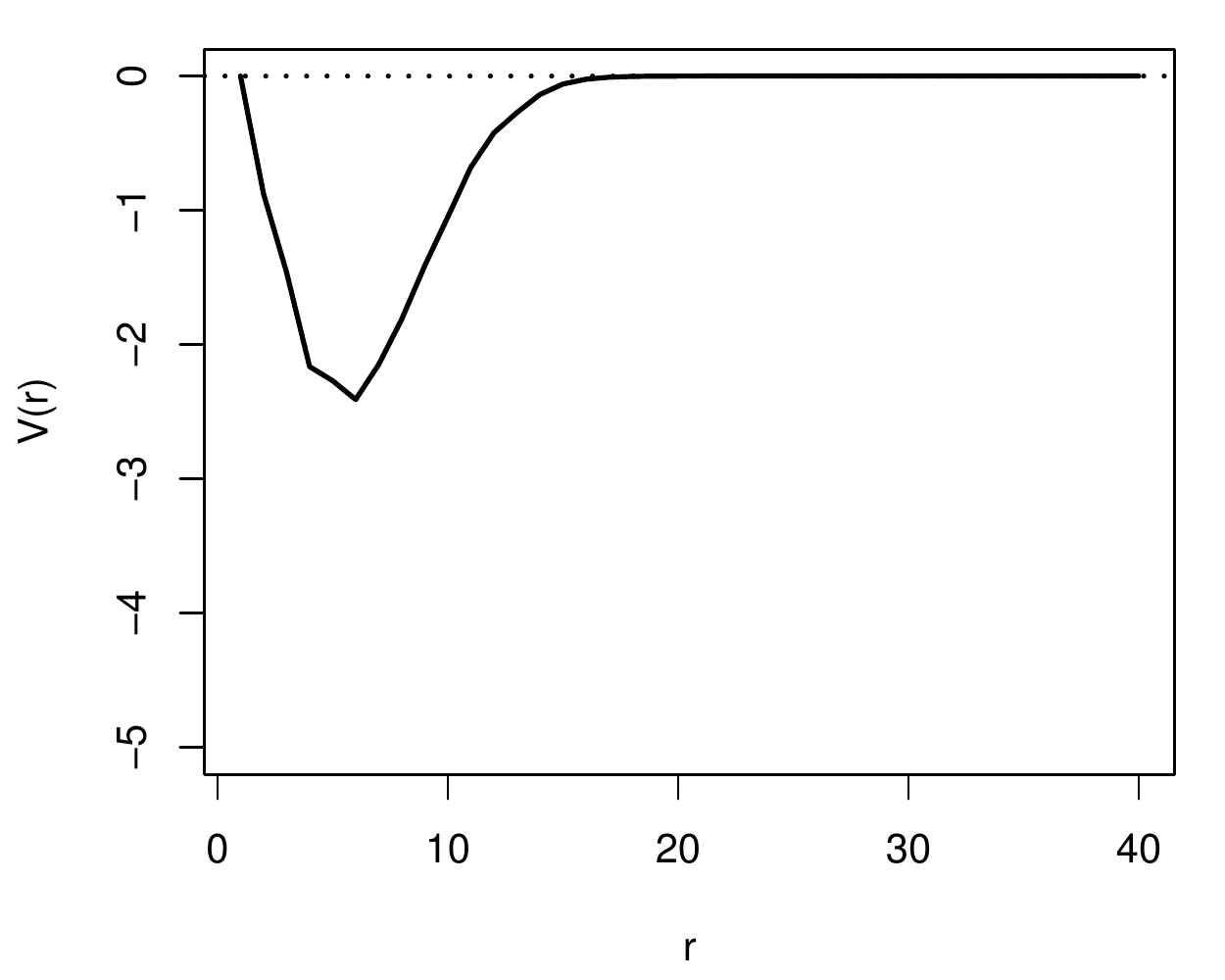}}}
\subfigure[$T(r)$ vs.~$r$, with $\theta_0=10$]{\scalebox{0.6}{\includegraphics{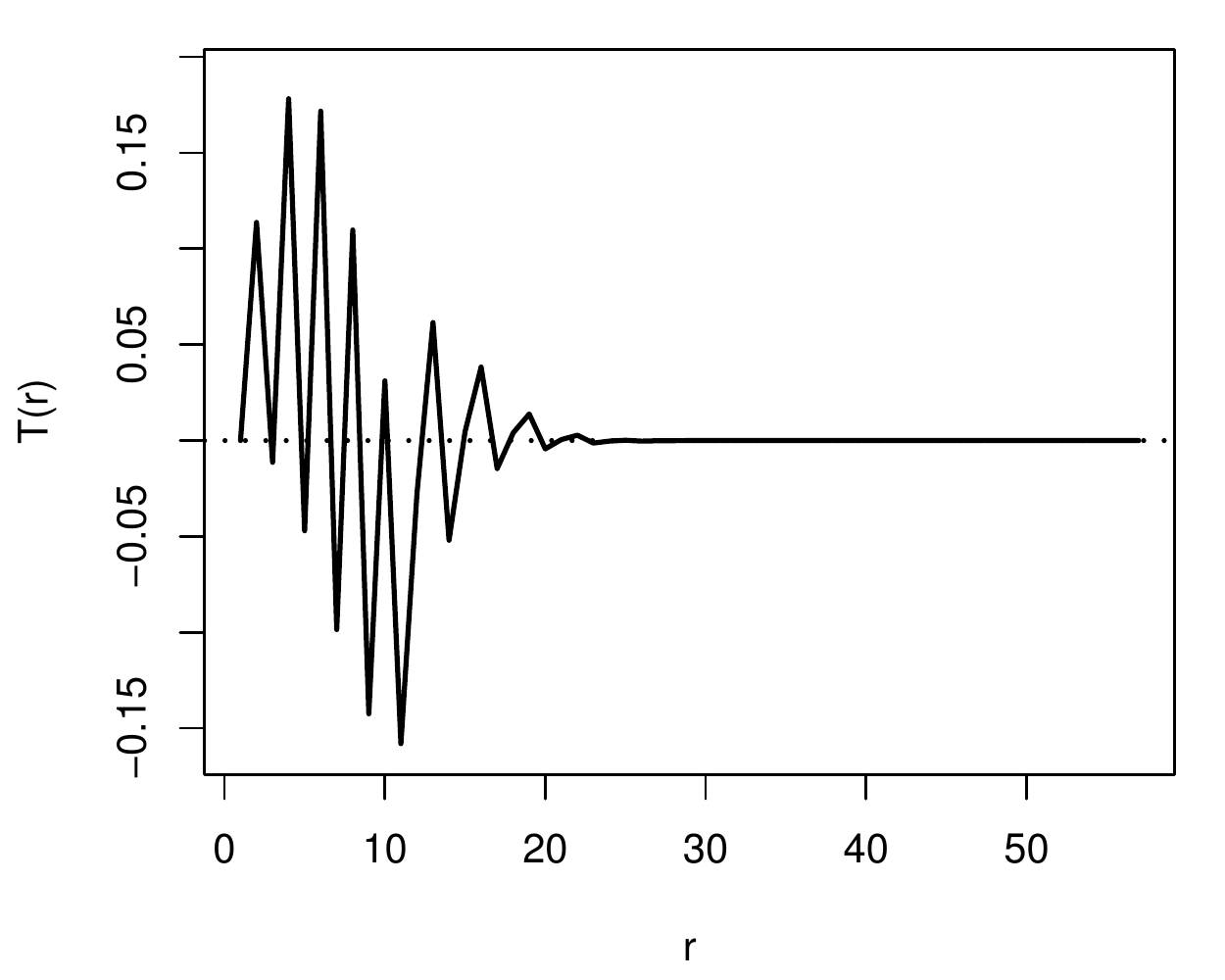}}}
\subfigure[$V(r)$ vs.~$r$, with $\theta_0=10$]{\scalebox{0.6}{\includegraphics{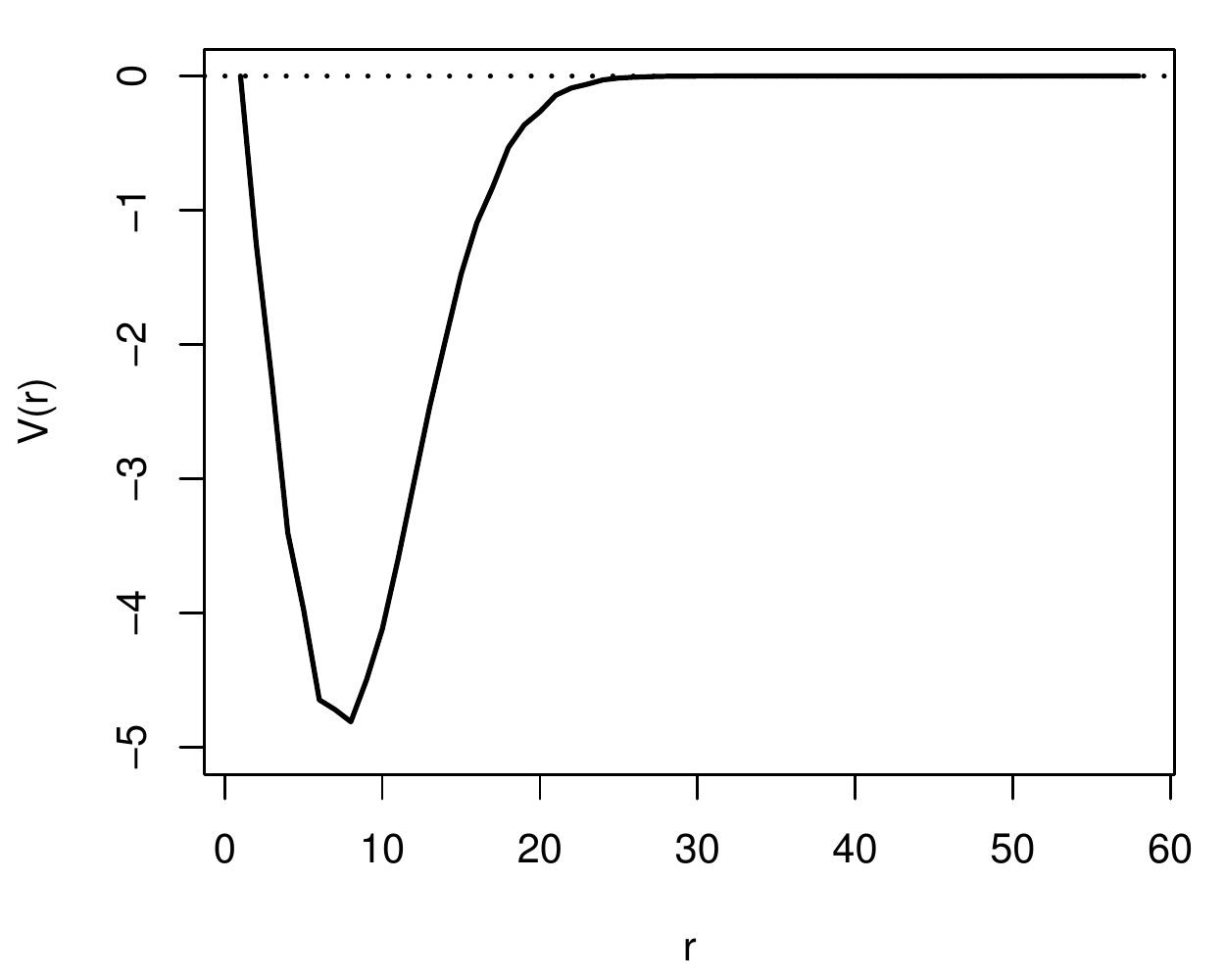}}}
\caption{Numerical checks that Algorithm~\ref{algo:sort} produces a ranking $\rho$ such that \eqref{eq:score.balance} and \eqref{eq:V.inequality} approximately hold.  Here $r$ is the index in Algorithm~\ref{algo:sort} and $T(r)$ and $V(r)$ are defined in \eqref{eq:TV.r}. The top row is for $\theta_0=5$ and the second row for $\theta_0=10$; the same vertical axis scale is used in both rows.}
\label{fig:checks}
\end{center}
\end{figure}

\subsection{Numerical illustrations---mean only}
\label{SS:mean}

Here we study the plausibility function $\pl_x(\theta_0; \S_\rho) = 1-\bel_x(\{\theta_0\}^c; \S_\rho)$ based on the optimal ranking $\rho=\rho^\star$ in Section~\ref{SS:recursive}.  The belief function at $\{\theta_0\}$ is zero for all $\theta_0$ so we can safely ignore it.  We will compare the plausibility function behavior to that of two classical textbook methods for testing $H_0: \theta=\theta_0$ versus $H_1: \theta\neq\theta_0$.  

\begin{enumerate}
\item \emph{Normal approximation}.  A naive approximation is to assume $X \sim \nm(\theta, \theta)$.  Then the textbook size-$\alpha$ normal test rejects $H_0$ based on observed $X=x$ iff $p_1(x;\theta_0) \equiv 2-2\Phi(\theta_0^{-1/2}|x-\theta_0|) \leq \alpha$.  In light of \eqref{eq:imtest}, we take $p_1(x; \theta_0)$ as the ``plausibility function'' corresponding to this normal test procedure.  
\vspace{-2mm}
\item \emph{Poisson equal-tail approximation}.  A somewhat less-naive size-$\alpha$ test rejects $H_0$ based on observed $X=x$ iff $F_{\theta_0}(x) \leq \alpha/2$ or $1-F_{\theta_0}(x-1) \leq \alpha/2$.  Equivalently, this test rejects $H_0$ iff $p_2(x;\theta_0) \equiv 2\min\{F_{\theta_0}(x), 1-F_{\theta_0}(x-1)\} \leq \alpha$.  We take $p_2(x; \theta_0)$ as the ``plausibility function'' corresponding to this test procedure. 
\end{enumerate}

Figure~\ref{fig:pl.cdf} shows the distribution functions of $p_1(X;\theta_0)$, $p_2(X;\theta_0)$, and $\pl_X(\theta_0; \S_\rho)$, all treated as functions of the random variable $X \sim \pois(\theta)$, for a variety of $\theta$ values, with $\theta_0=7$.  There are two things to look for in these plots.  The first, for $\theta=\theta_0$, is that the distribution function does not exceed the diagonal line corresponding to the distribution function of $\unif(0,1)$.  This demonstrates the validity property.  In Panel~(c) we find that only the IM-based plausibility function satisfies the validity criterion.  The second thing we are looking for is stochastic dominance.  Specifically, if one distribution function is uniformly smaller than another distribution function, then the former corresponding plausibility function is stochastically larger than the latter.  This, in turn, means that inference based on the former will, in general, be more efficient.  Panels~(a) and (b) show no clear dominance, but the IM tends to outperform the normal approximation.  Panels~(d)--(f) show that the IM-based plausibility function dominates, stochastically, the other two and, hence, the corresponding inference is more efficient.  

% One that sits above the IM CDF is the equal-tail Poisson plausibility...

\begin{figure}%[t]
\begin{center}
\subfigure[$\theta=4$]{\scalebox{0.6}{\includegraphics{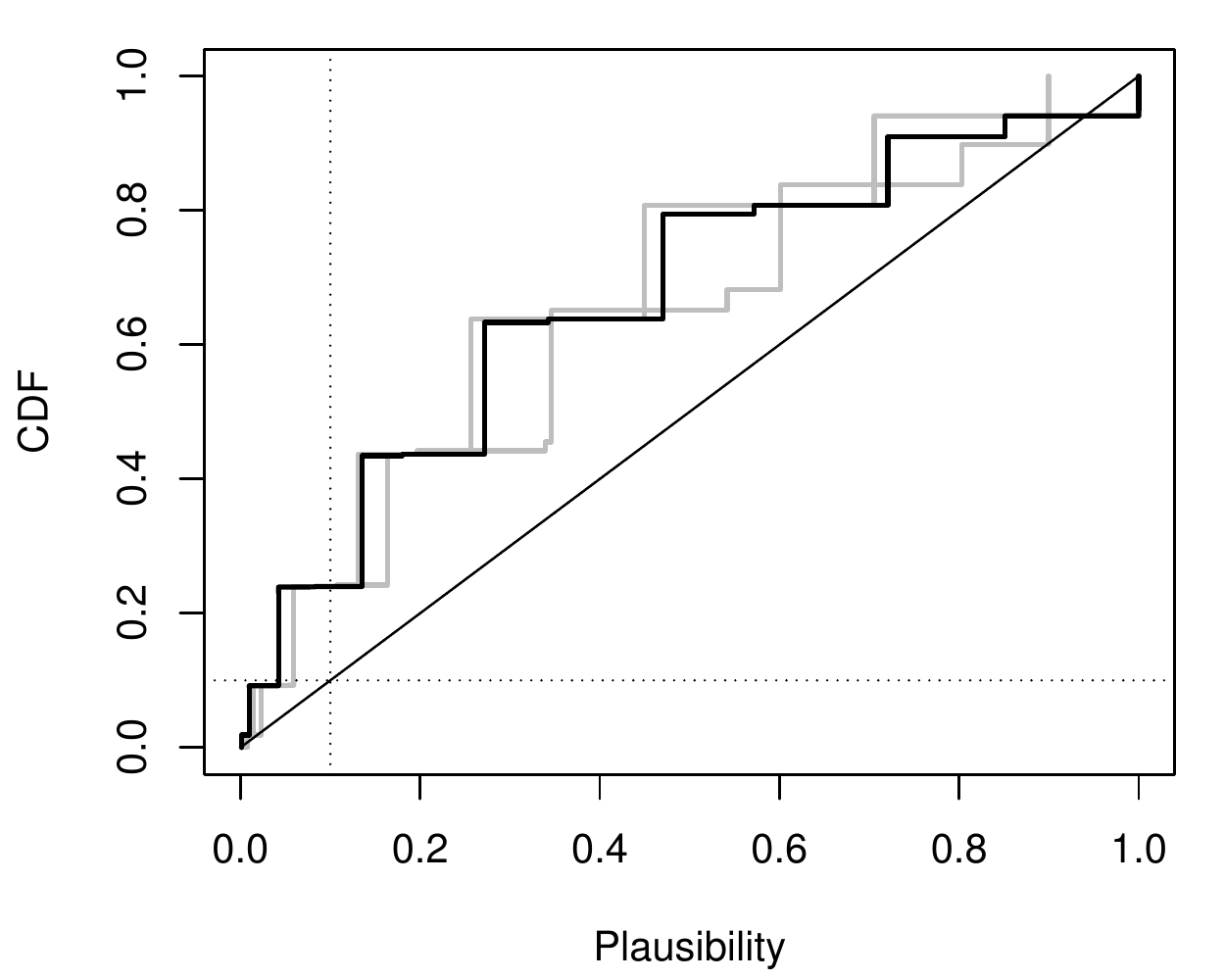}}}
\subfigure[$\theta=6$]{\scalebox{0.6}{\includegraphics{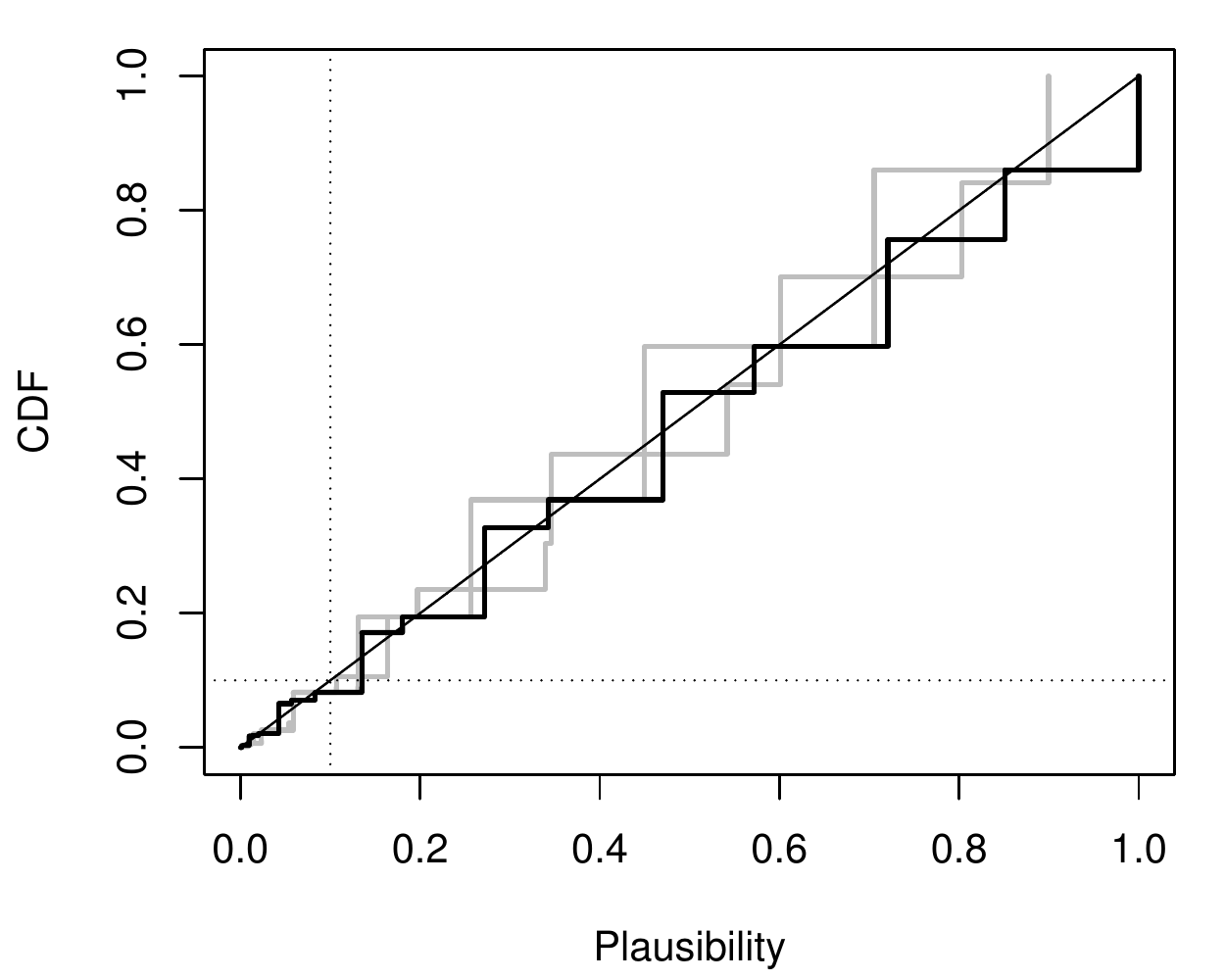}}}
\subfigure[$\theta=\theta_0=7$]{\scalebox{0.6}{\includegraphics{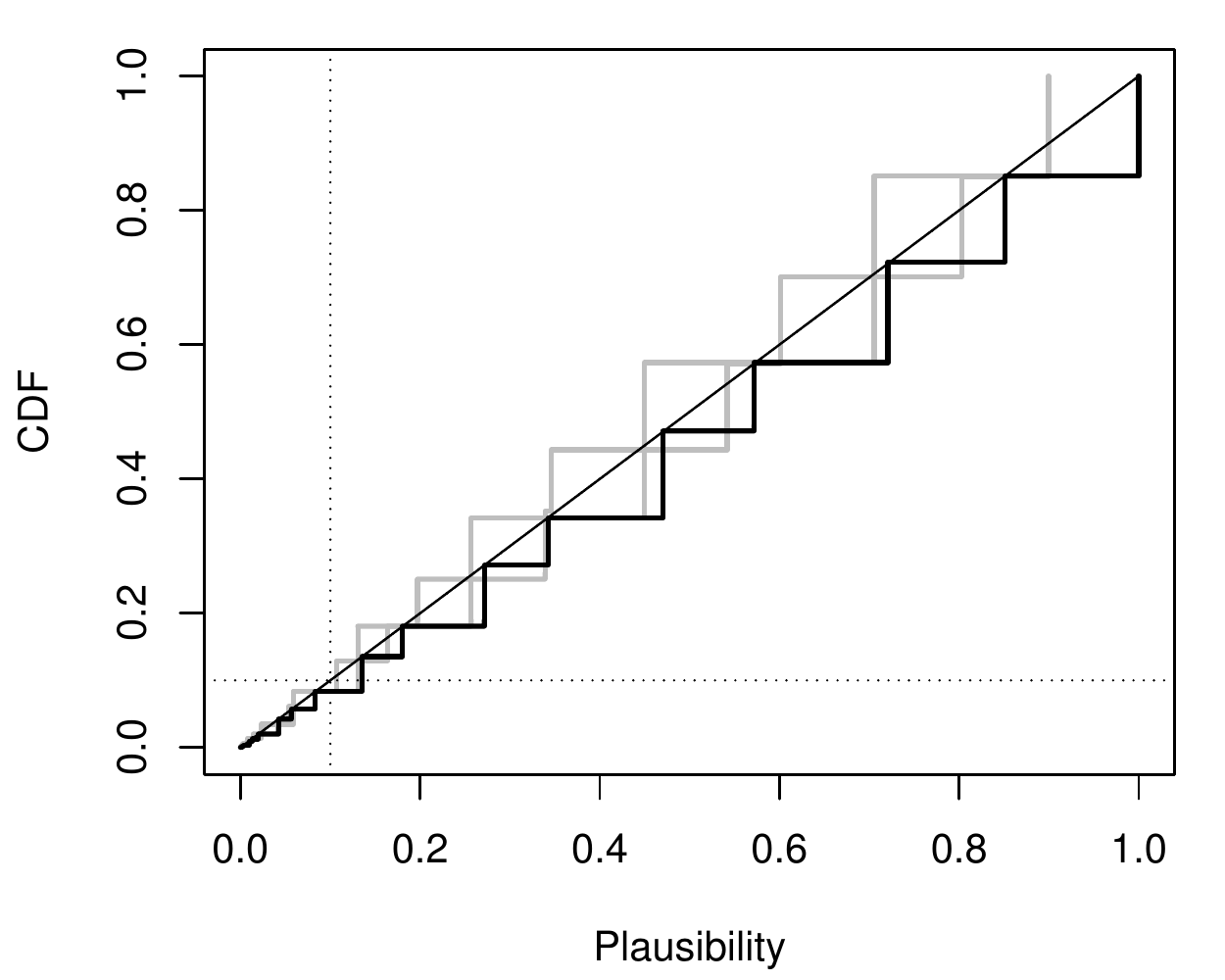}}}
\subfigure[$\theta=8$]{\scalebox{0.6}{\includegraphics{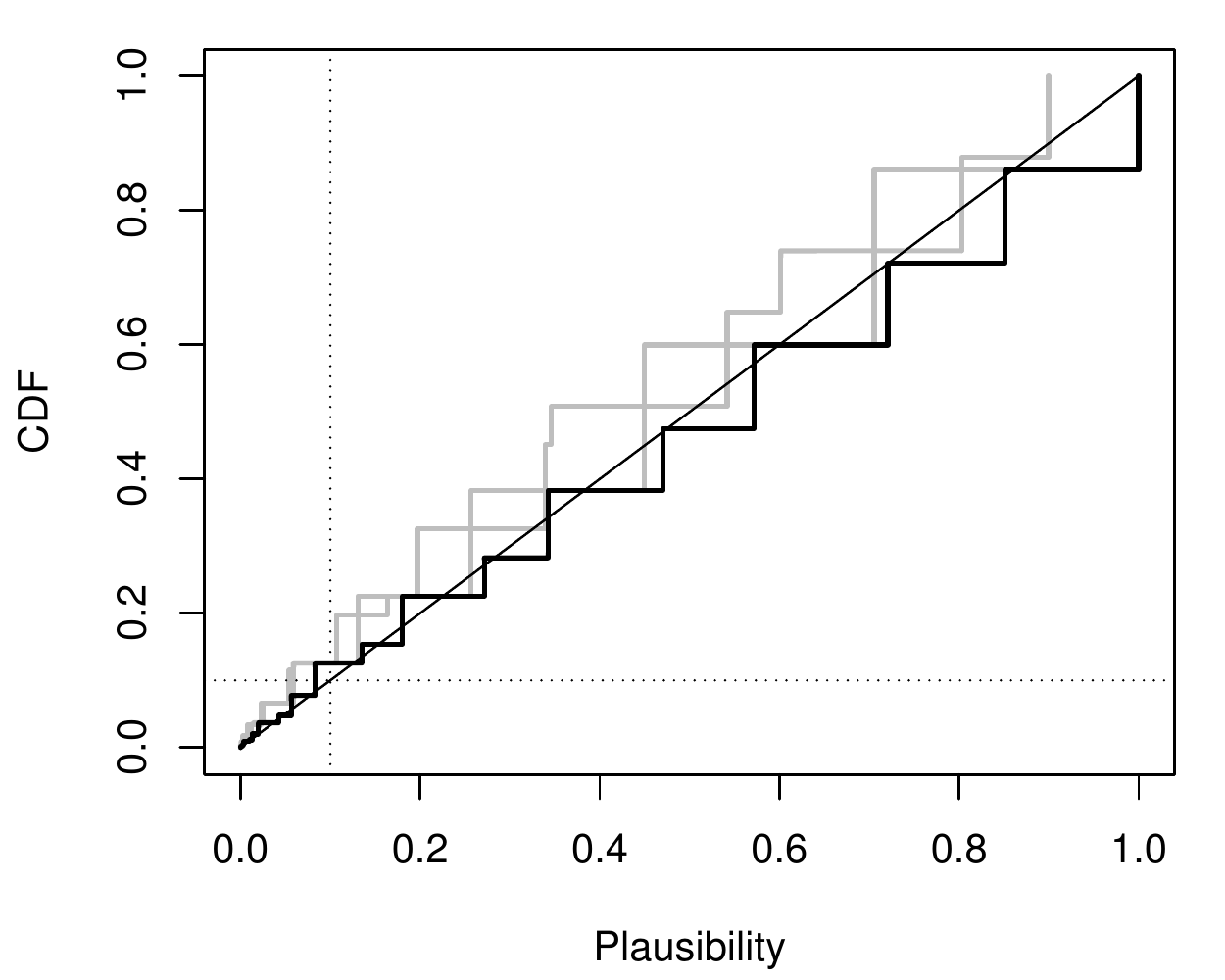}}}
\subfigure[$\theta=10$]{\scalebox{0.6}{\includegraphics{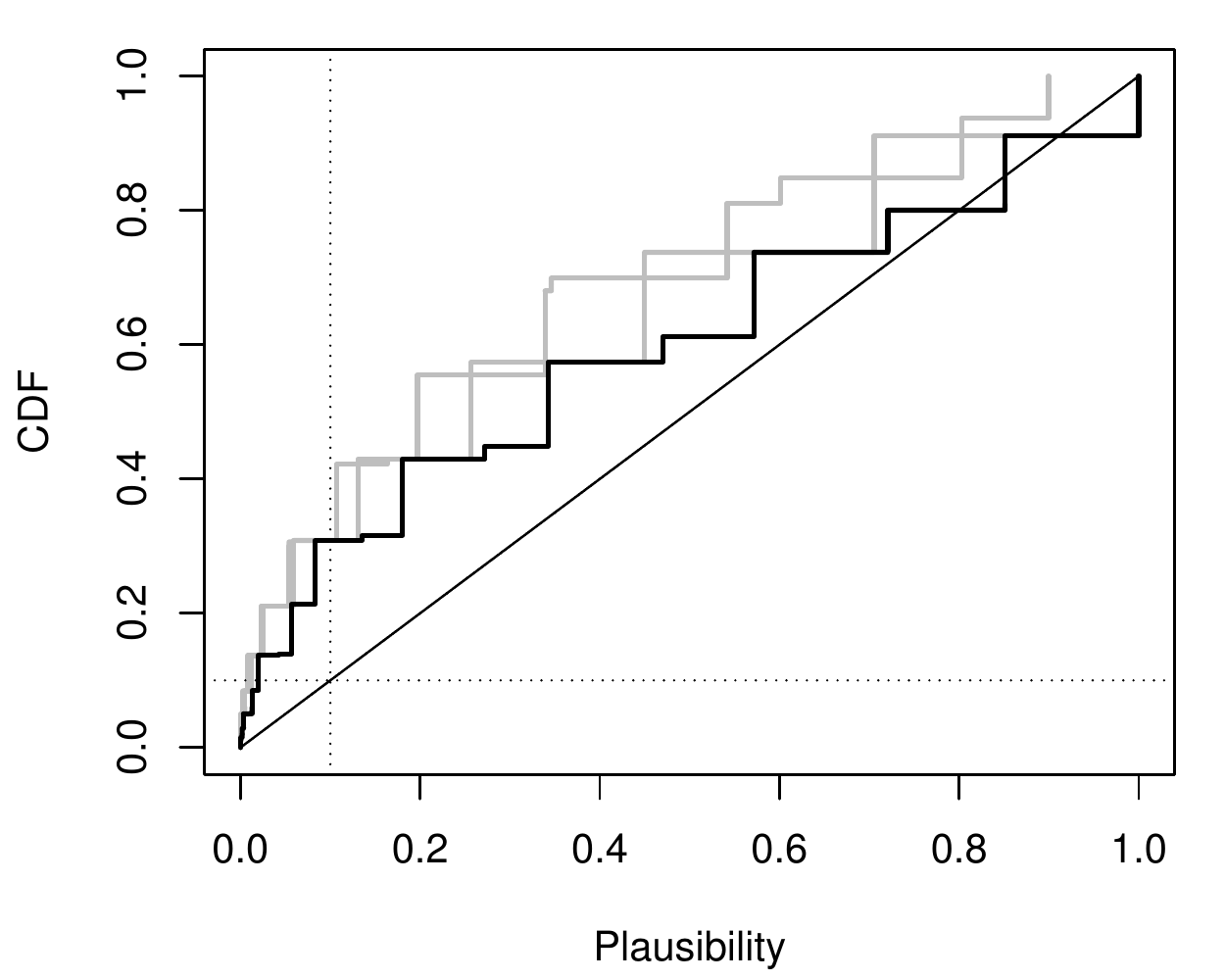}}}
\subfigure[$\theta=12$]{\scalebox{0.6}{\includegraphics{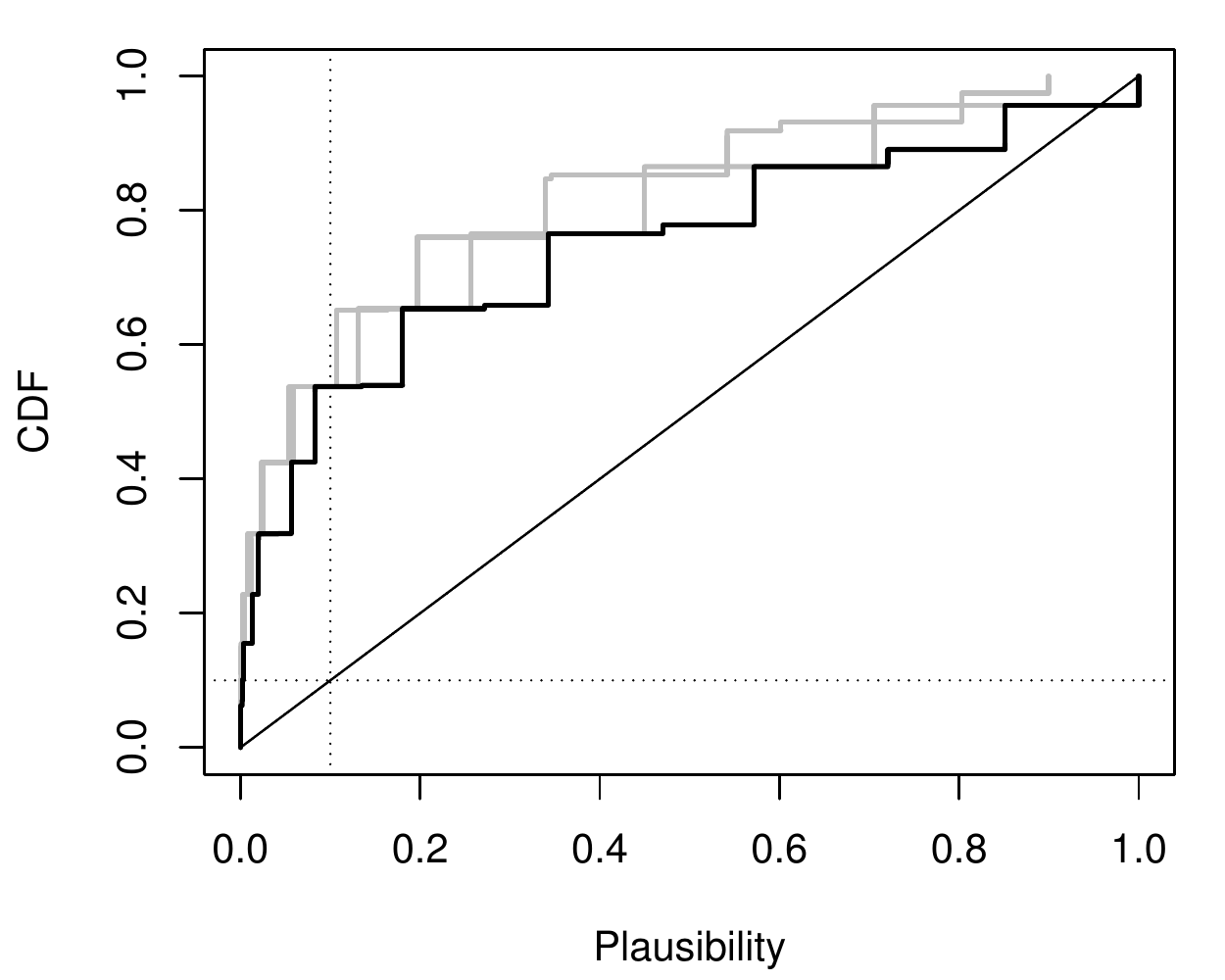}}}
\caption{Plots of the distribution function (CDF) of $\pl_X(\theta_0)$, when $X \sim \pois(\theta)$, for $\theta_0=7$ and various $\theta$'s.  In each panel, the two gray lines correspond to the two ``frequentist plausibility functions'' described in the text; the black line corresponds to the optimal IM plausibility function.  Each is based on 100,000 Monte Carlo samples.}
\label{fig:pl.cdf}
\end{center}
\end{figure}

Figure~\ref{fig:pois.pl} plots the ``plausibility functions'' $p_1(x;\theta)$ and $p_2(x;\theta)$, based on the frequentist methods, along with the optimal IM plausibility function, as functions of $\theta$ for various $x$ values.  One general observation is that both the IM and the normal plausibility functions peak at $\theta=x$, the maximum likelihood estimate, shown by a vertical line, while the Poisson equal-tail plausibility function is off-center.  The horizontal line describes the $\alpha=0.1$ level sets, i.e., the 90\% plausibility intervals.  In each case, the normal plausibility interval---which corresponds exactly to the textbook confidence interval---is a hair shorter than the IM plausibility interval.  However, unlike the IM plausibility interval, which has coverage guarantees via the validity theorem (see Panel~(c) of Figure~\ref{fig:pl.cdf}), the normal confidence interval has no such guarantees in this sort of mis-specified model.  

\begin{figure}[t]
\begin{center}
\subfigure[$x=0$]{\scalebox{0.6}{\includegraphics{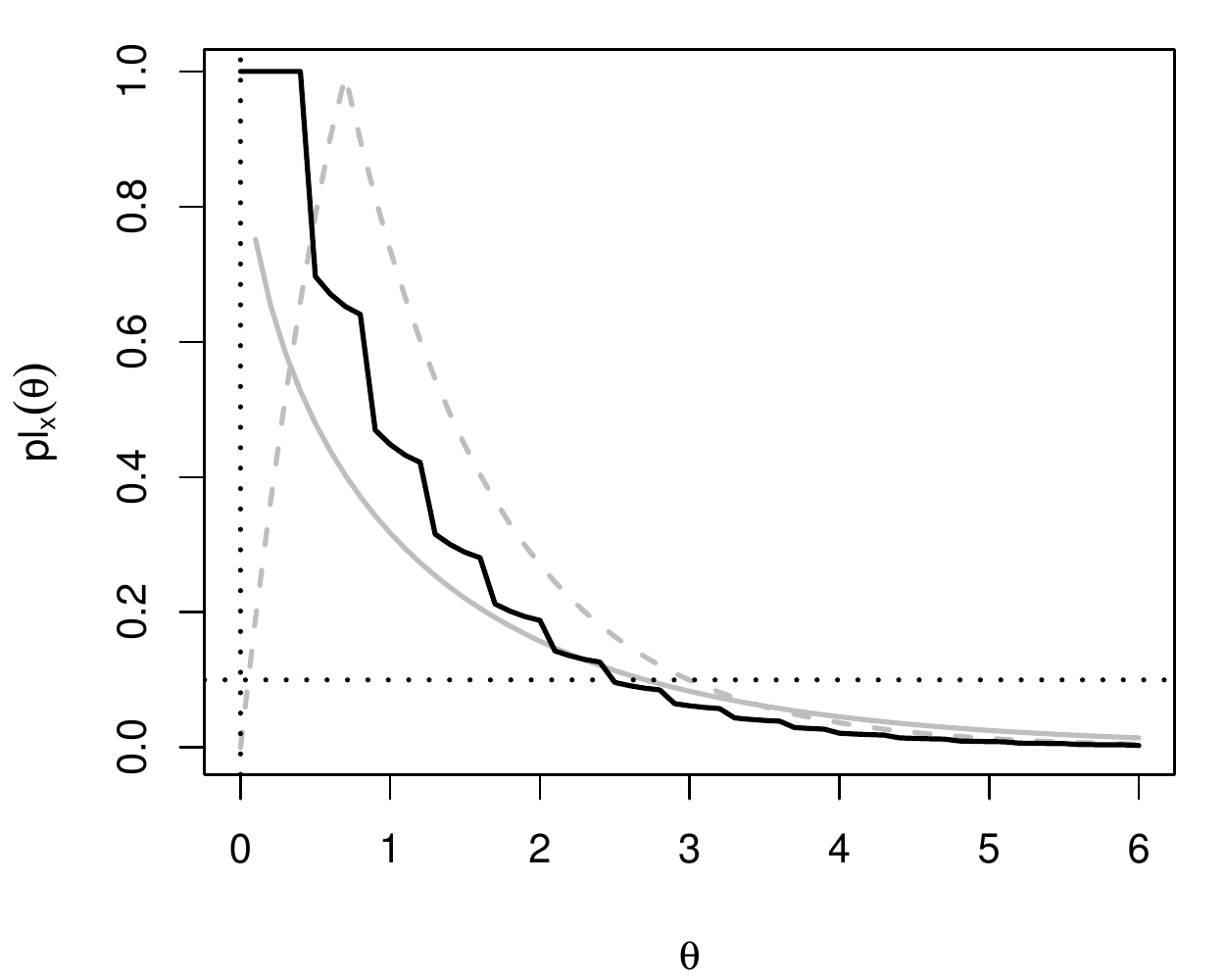}}}
\subfigure[$x=3$]{\scalebox{0.6}{\includegraphics{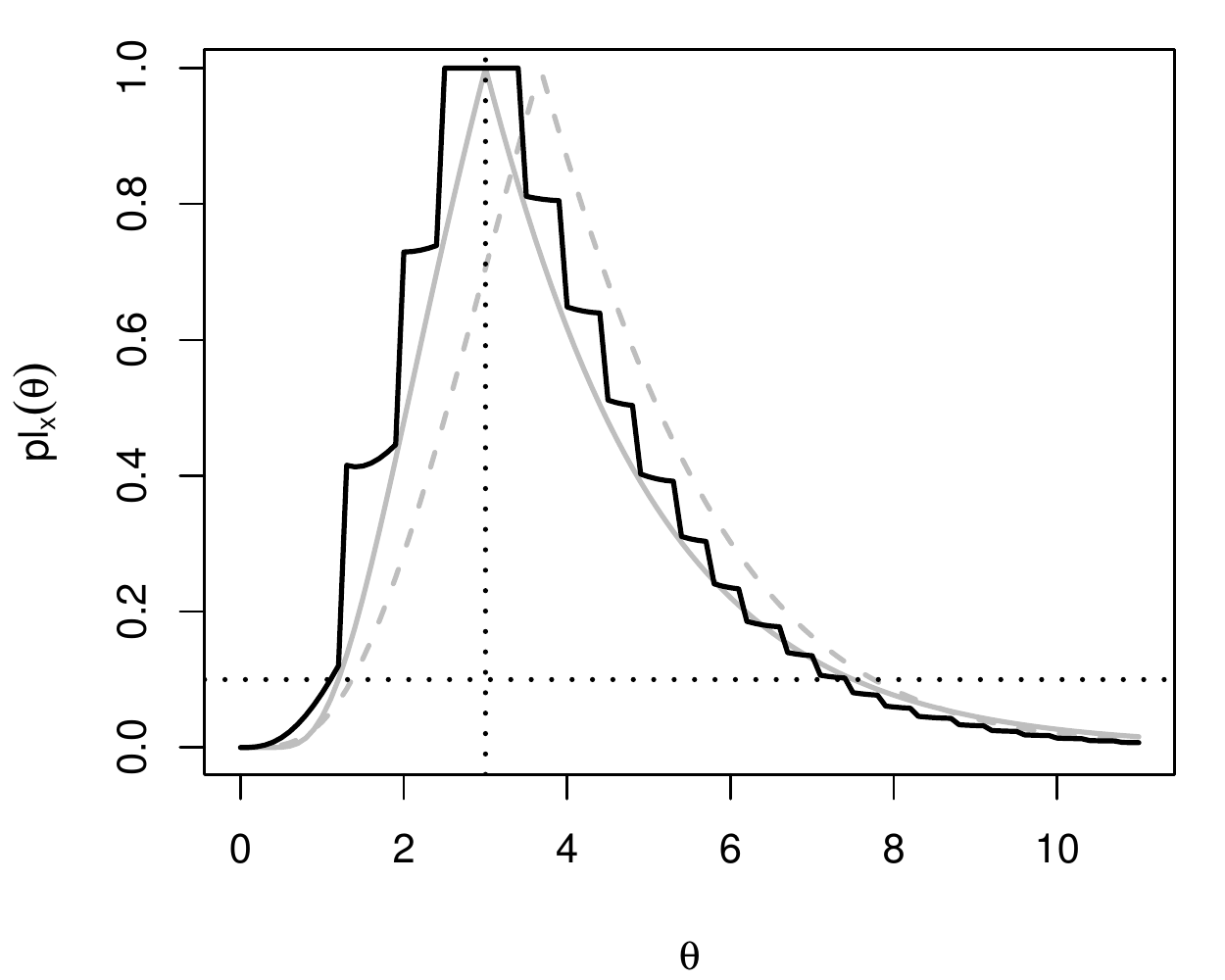}}}
\subfigure[$x=7$]{\scalebox{0.6}{\includegraphics{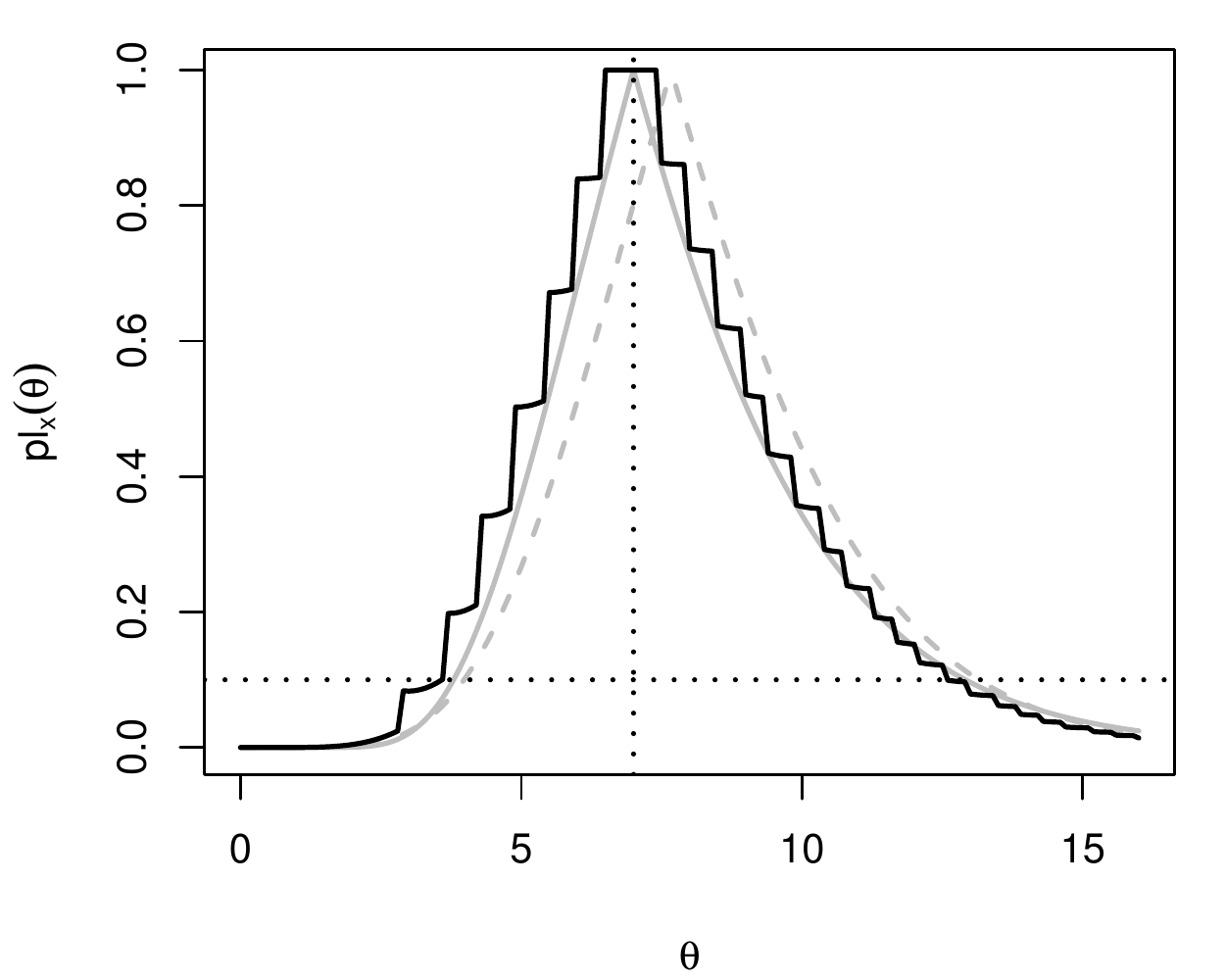}}}
\subfigure[$x=10$]{\scalebox{0.6}{\includegraphics{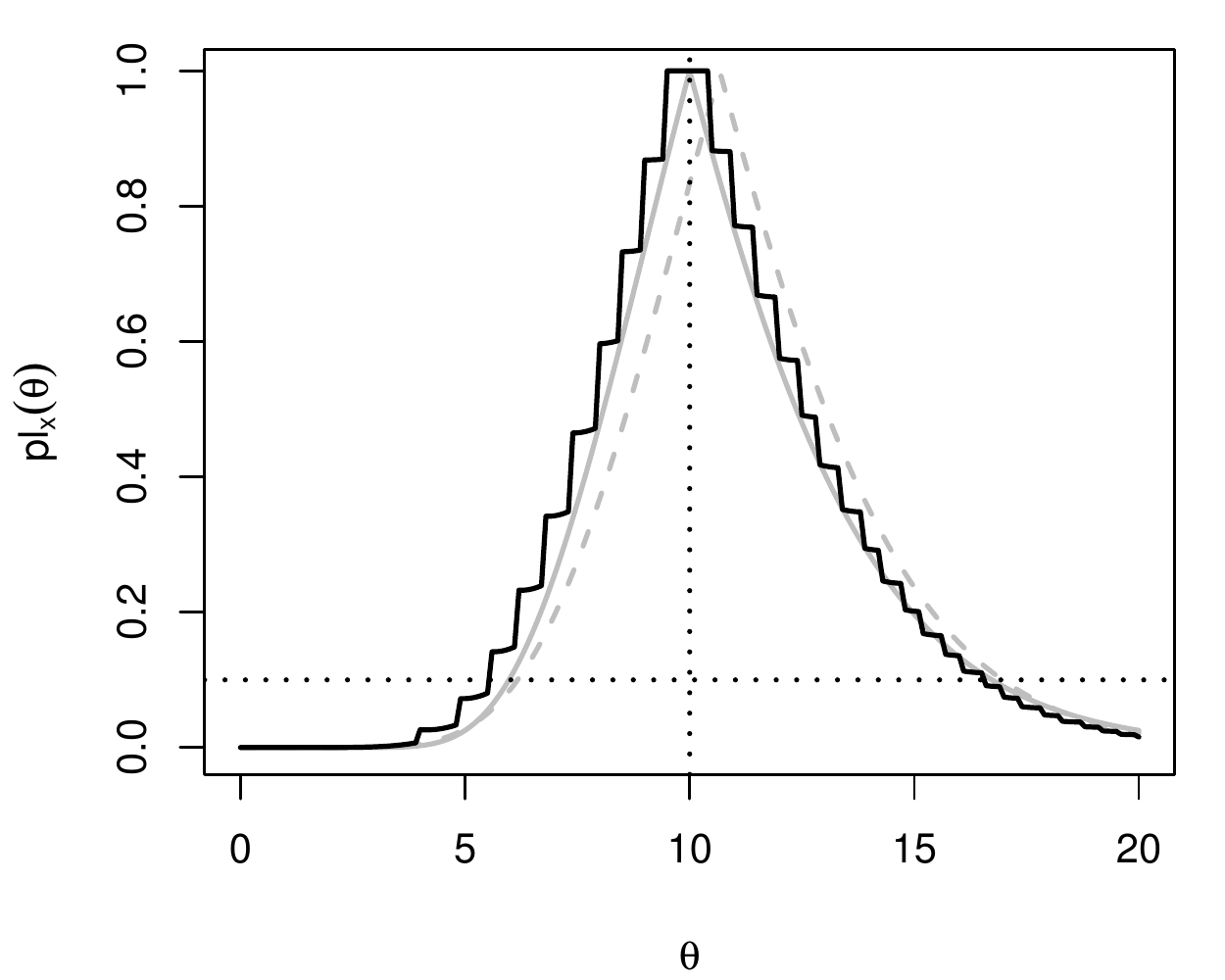}}}
\caption{Plots of $\pl_x(\theta)$, as a function of $\theta$, for various $x$ values.  In each panel, solid and dashed gray lines are ``plausibility functions'' $p_1(x;\theta)$ and $p_2(x;\theta)$, respectively, and the the solid black line is the optimal IM plausibility function.}
\label{fig:pois.pl}
\end{center}
\end{figure}

\subsection{Numerical illustrations---mean plus background}
\label{SS:background}

As shown above, the discreteness of the Poisson random variable makes it challenging to develop an efficient IM for its mean, $\theta$.  An additional challenge arises when one considers an \emph{a priori} constraint on the possible values of $\theta$.  An IM for the constrained Poisson mean was developed in \citet{leafliu2012}.  Here, we briefly review the problem and then introduce a more efficient IM using the scheme in Section~\ref{SS:recursive}.  Several frequentist methods have also been developed for this problem; see \citet{mandelkern2002}.

Suppose for example that the Poisson count, $X$, is comprised of a number of signal events, $S$, and independent background events, $B$, so that $X = S + B$.  If $S \sim \pois(\lambda)$ and $B \sim \pois(\beta)$, then $X \sim \pois(\lambda + \beta)$.  Now suppose that the value of $\beta$ has been established with certainty.  If $\theta = \lambda + \beta$ is the mean of $X$, then the fact that $\lambda$ must be nonnegative implies the constraint, i.e., $\theta \geq \beta$.  The problem with ignoring such a constraint is clear in Figure~\ref{fig:pois.pl}---$\pl_x(\theta)$ can be positive for any $\theta$ value, even for $\theta < \beta$.  So, in light of the constraint $\theta \geq \beta$, the IM must be modified appropriately.

Technically, the problem can be seen in the auxiliary variable, $U$.  After observing $x$, the constraint implies that $U$ must lie in a strict subset of $\UU$.  Applying the constraint, $\theta \in [\beta, \infty)$, to \eqref{eq:poisson.association}, leads to a constraint on $U$: $G_{x+1}(\beta) < u \leq 1$.  Without considering constraints, $\S$ is intended to predict $U$ realizations anywhere in $\UU$.  Some members of its support $\SS$ may not be contained in $(G_{x+1}(\beta),1]$; these are conflict cases.  Let $S'$ be the largest $S \in \SS$ such that $S \cap (G_{x+1}(\beta), 1] = \varnothing$.  The probability on $S'$ and all its subsets is known as conflict mass: $\prob_U\{S'\} = \bel_x([\beta, \infty)^c; \S)$.  An IM for the constrained $\theta$ must distribute this conflict mass somewhere in the constraint set.

The elastic belief method \citep{leafliu2012} expands conflict cases so that each one intersects with the constraint.  In effect, the conflict mass is moved to a subset of the parameter constraint set.  The proof of validity for the elastic belief method also applies to more general procedures.  Therefore, it is not necessary to formulate the mathematical details of the elastic belief method in this problem.
%In fact, validity holds when the conflict mass is moved to any arbitrary subset of the parameter constraint.  
We can simply place any conflict mass on $\{\beta\}$, which is on the boundary of the constraint.  The resulting plausibility function for point assertions is:
\[
\pl_x(\{\theta_0\}; \S_\rho' ) = \begin{cases}
0 & \text{if $\theta_0 < \beta$;} \\
1 & \text{if $\theta_0 = \beta$ and $\bel_x([\beta, \infty)^\mathrm{c}; \S_\rho) > 0$;} \\
\pl_x(\{\theta_0\}; \S_\rho) & \text{otherwise,}
\end{cases}
\]
where $\S_\rho'$ is the predictive random set implied by moving conflict cases to $\{\beta\}$, and $\S_\rho$ is the predictive random set constructed recursively in Section~\ref{SS:recursive}.  In the comparisons that follow, we refer to this as the EB--SB method, for elastic belief + score-balance.  The 90\% EB--SB plausibility interval, when $\beta = 15$, is shown as black lines in Figure~\ref{fig:pois_interval2}.  The gray lines correspond to the plausibility intervals in \citet{leafliu2012}.  EB--SB produces a shorter interval at each $x$ in the figure.  

%For comparison, Table \ref{tab:bothConstrMethods} shows the interval bounds and widths for $x \in [0,12]$ when $\beta=3$.

\begin{figure}[t]%[htbp]
\begin{center}
\includegraphics[scale=0.62]{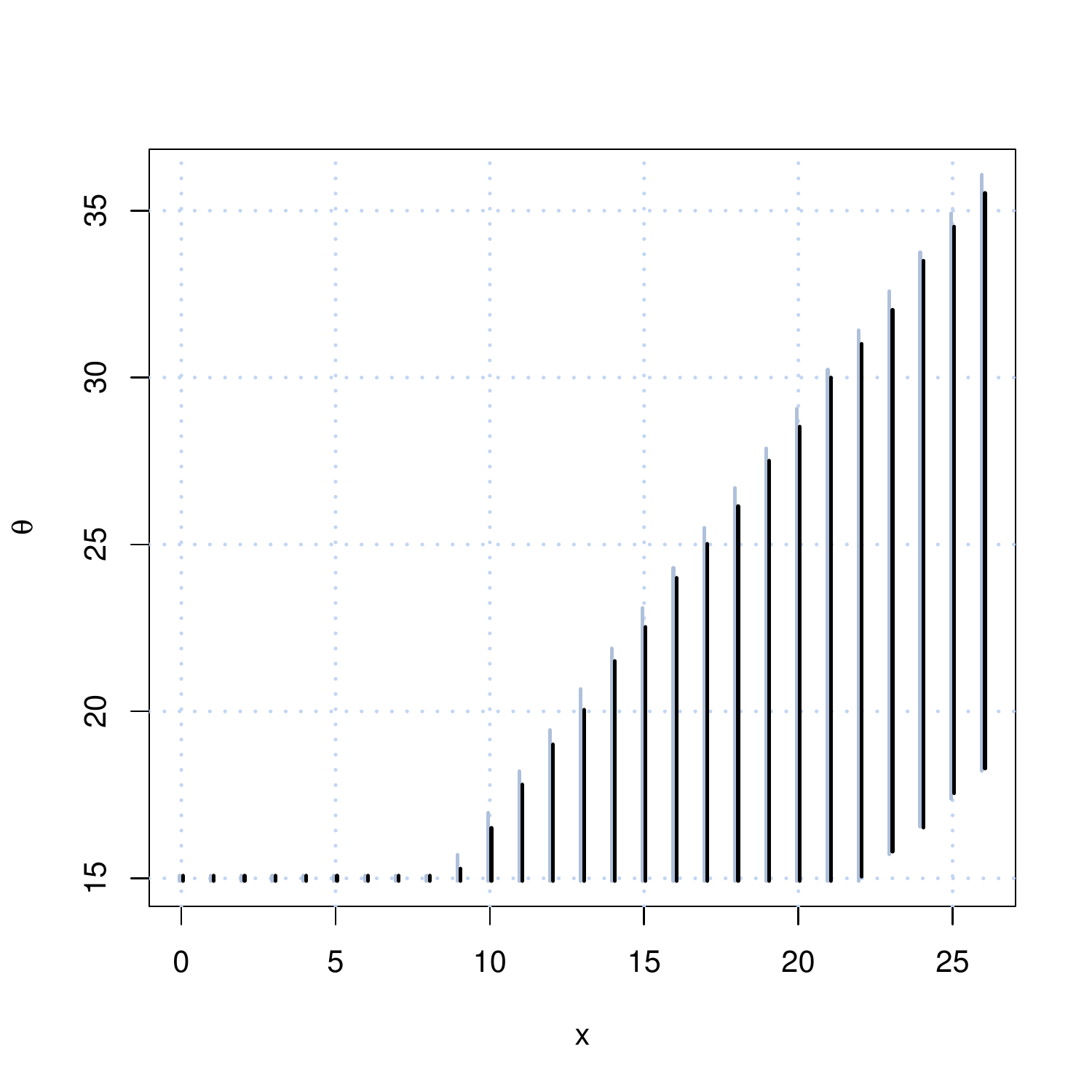}
\caption{90\% plausibility intervals for $\theta$ with $\beta=15$.  The black and gray lines are the intervals based on EB--SB and the method in \citet{leafliu2012}, respectively.}
\label{fig:pois_interval2}
\end{center}
\end{figure}

For further comparison, we consider a variety of existing methods: confidence intervals of \citet[][FC98]{Feldman.Cousins.1998}, \citet{Giunti1999}, \citet[][MS00b]{Mandelkern.Schultz.2000b}, \citet[][RW00]{Roe.Woodroofe.2000}, \citet{Roe.Woodroofe.1999} with the \citet{Mandelkern.Schultz.2000a} adjustment (RW+MS00a), and the plausibility interval of \citet[][ELL12]{leafliu2012}.  Figure~\ref{fig:coverage} shows the coverage probabilities for each interval estimate of $\lambda$, for $\beta=3$, as a function of $\lambda \in [0,4]$.  EB--SB seems to be the best performer in the left-hand column but, in the right-hand column, there is no clear winner.  Figure~\ref{fig:intervalWidths} plots the width of the nominal 90\% interval estimates, as a function of data $x$, with $\beta=3$, for the various methods described above.  Here we see that the EB--SB plausibility interval is the narrowest up to $x = 5$ at which point it becomes slightly wider than the intervals of other methods.  But we must reiterate: IMs are more than just tools to construct frequentist procedures.  That said, it is remarkable that the EB--SB plausibility intervals are as good or better than its competitors based on frequentist criteria.

\begin{figure}%[t]
\begin{center}
\subfigure[EB--SB vs.~FC98]{\scalebox{0.62}{\includegraphics{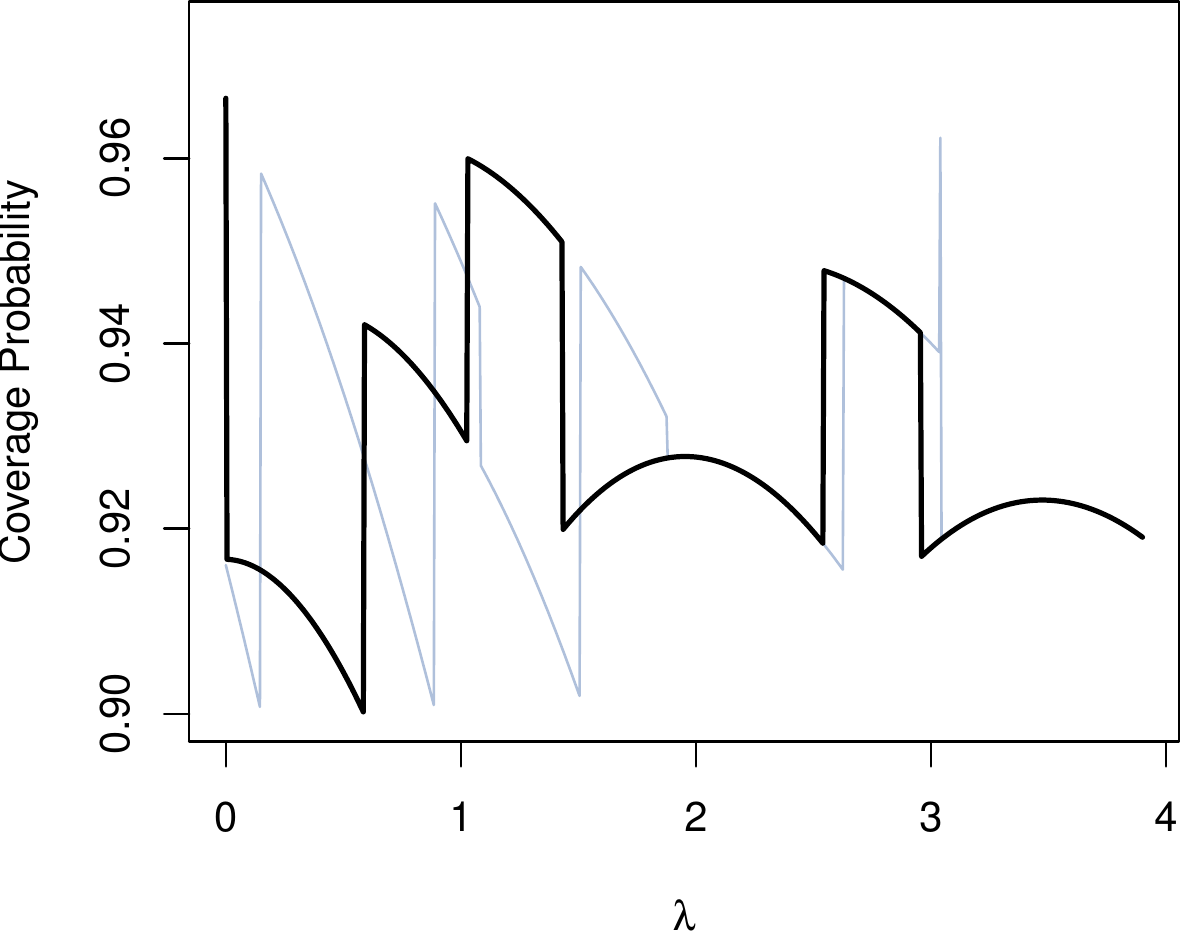}}}
\subfigure[EB--SB vs.~Giunti99]{\scalebox{0.62}{\includegraphics{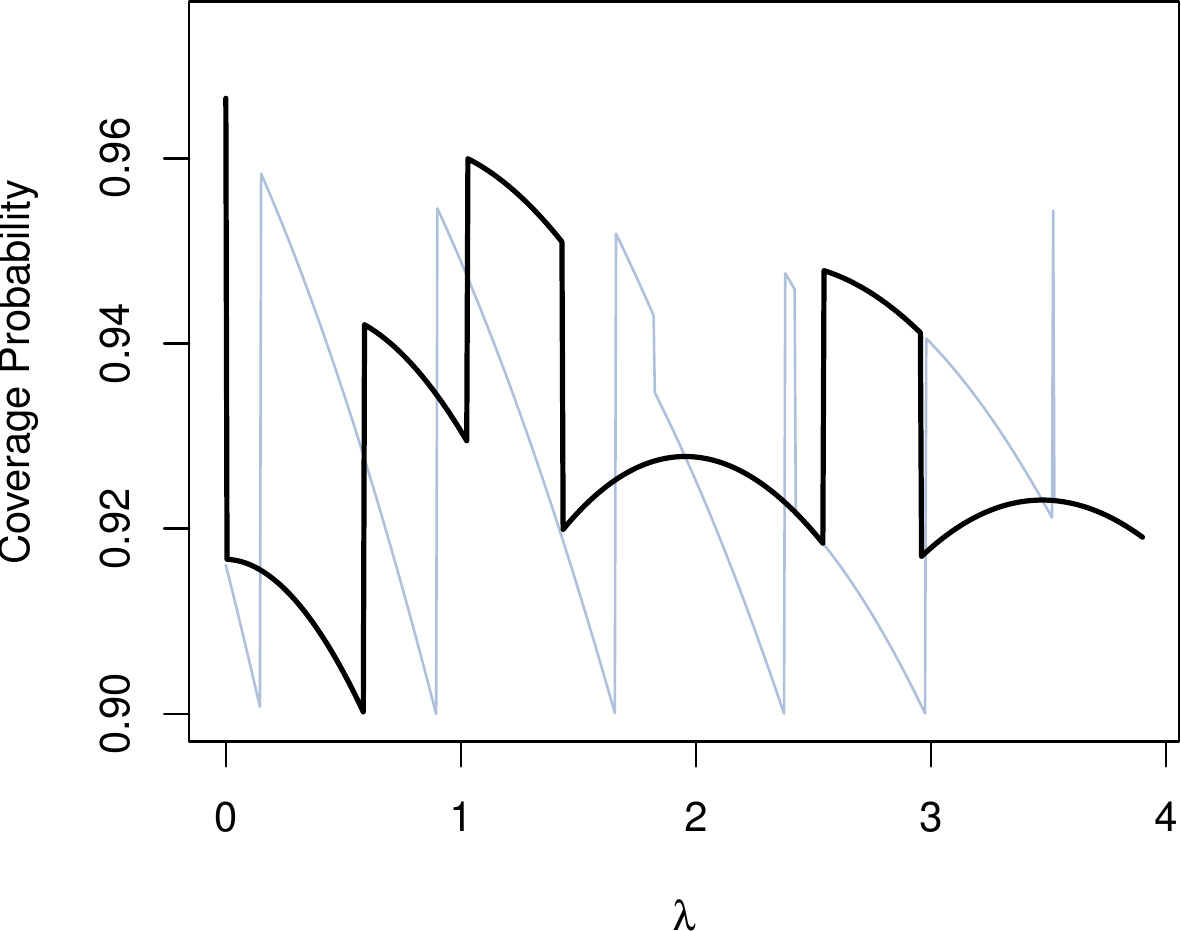}}}
\subfigure[EB--SB vs.~ELL12]{\scalebox{0.62}{\includegraphics{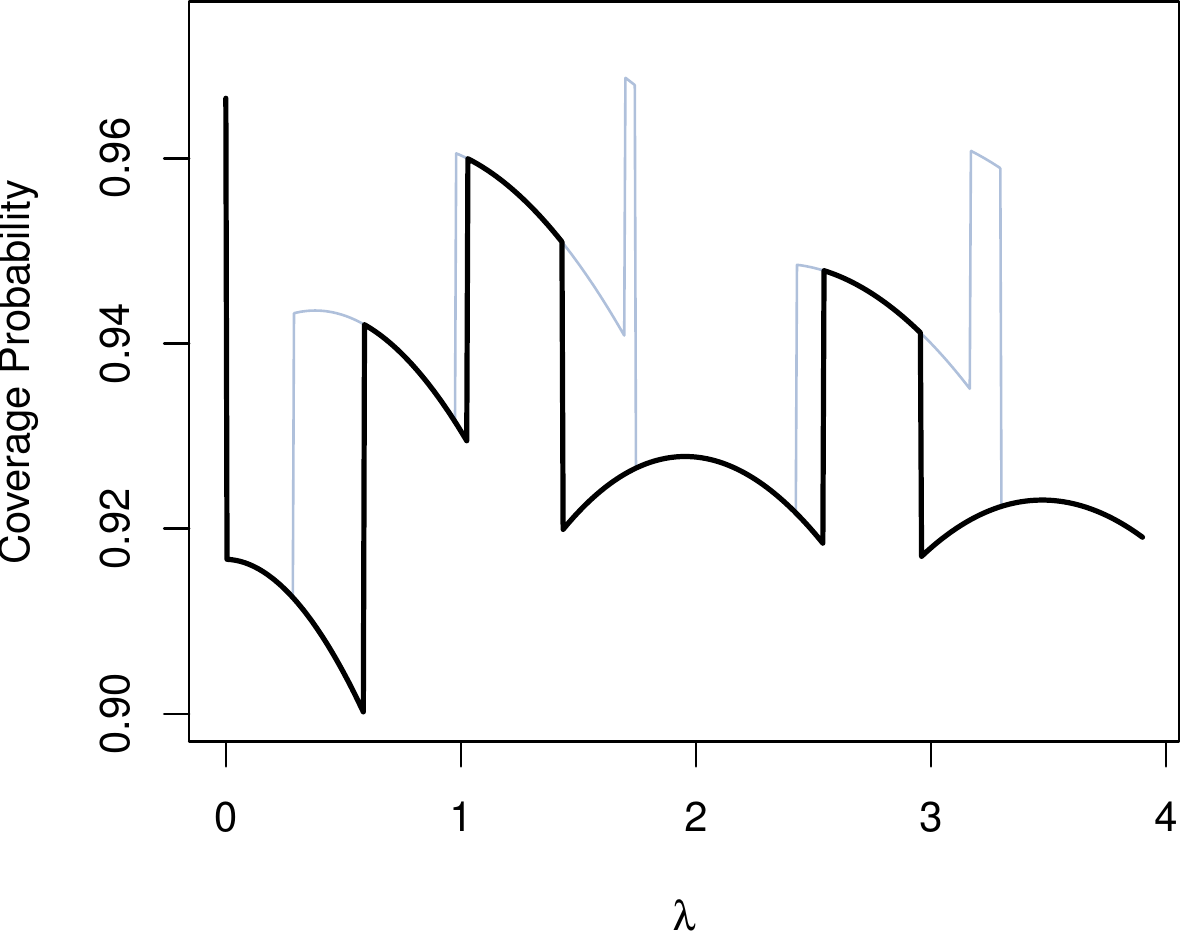}}}
\subfigure[EB--SB vs.~MS00b]{\scalebox{0.62}{\includegraphics{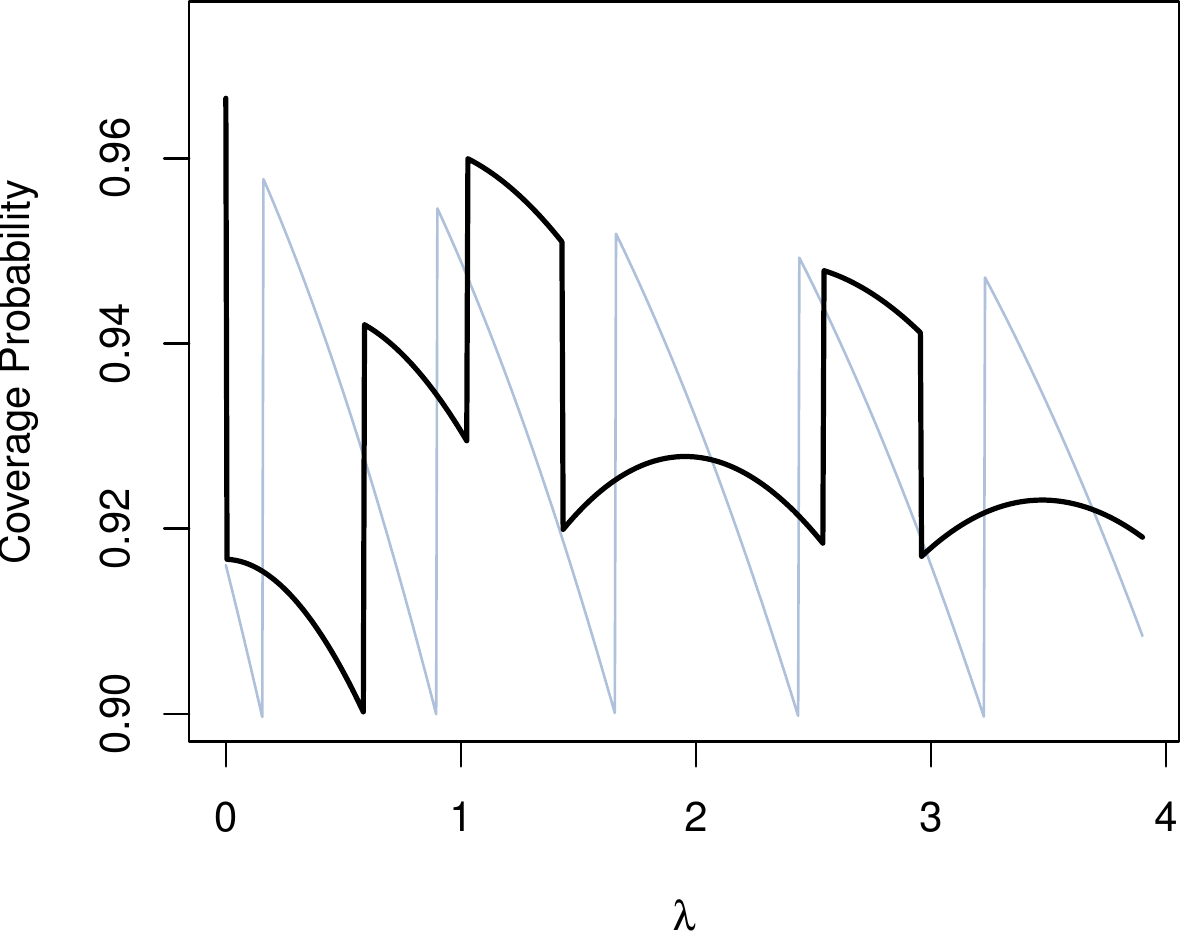}}}
\subfigure[EB--SB vs.~RW00]{\scalebox{0.62}{\includegraphics{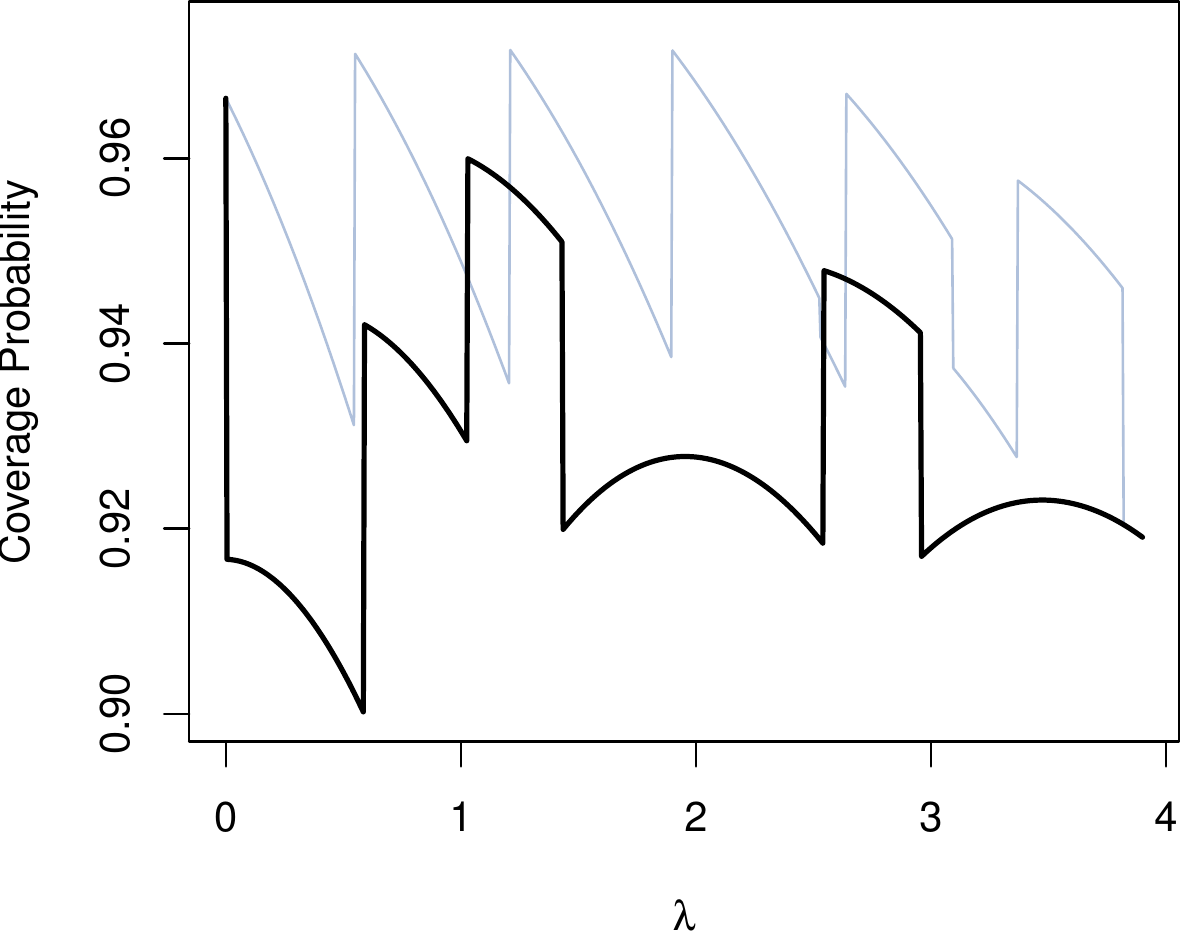}}}
\subfigure[EB--SB vs.~RW99+MS00a]{\scalebox{0.62}{\includegraphics{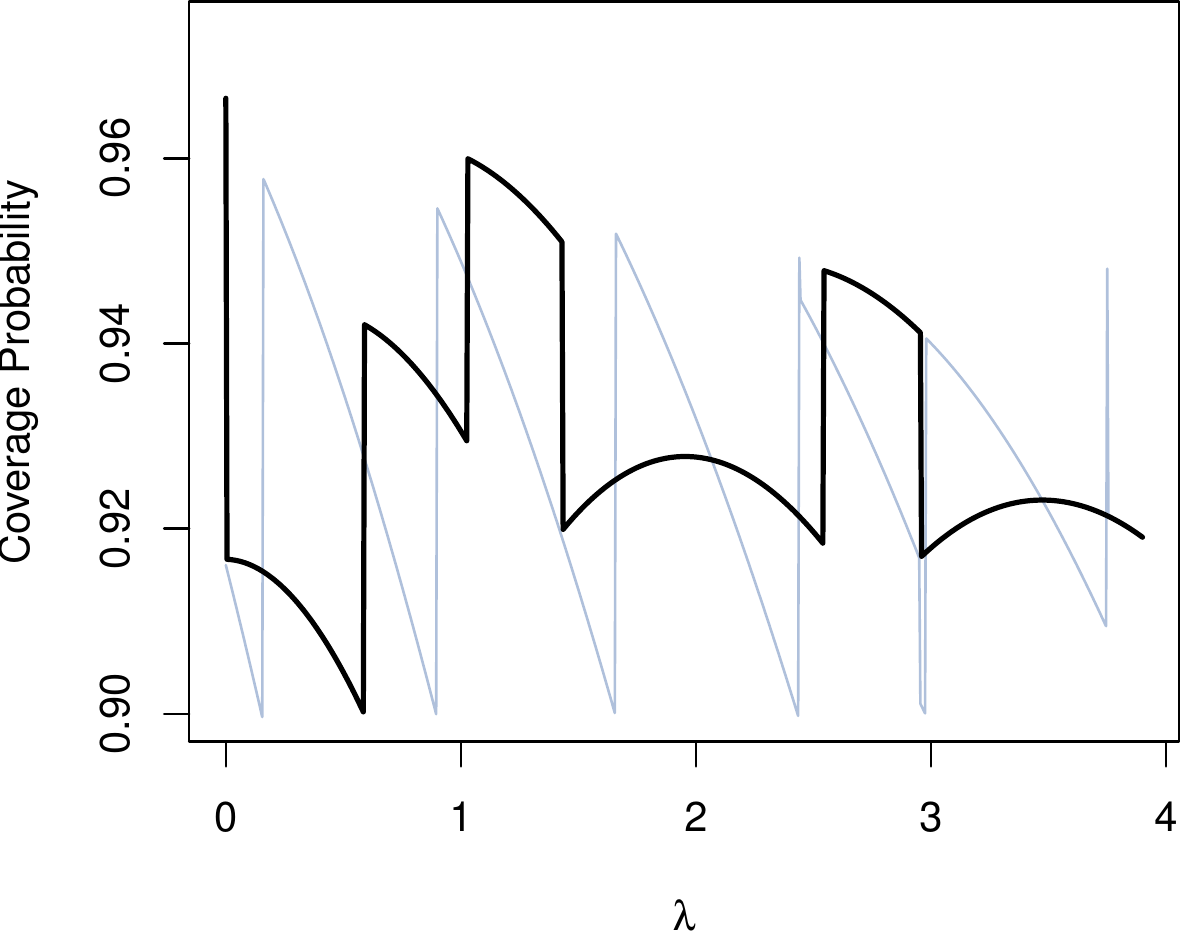}}}
\caption{Coverage probabilities comparisons for the nominal 90\% EB--SB plausibility intervals (black) against various confidence intervals (gray) for $\lambda \in [0,4]$, with $\beta=3$.}
\label{fig:coverage}
\end{center}
\end{figure}

\begin{figure}[htbp]
\begin{center}
\includegraphics[scale=0.75]{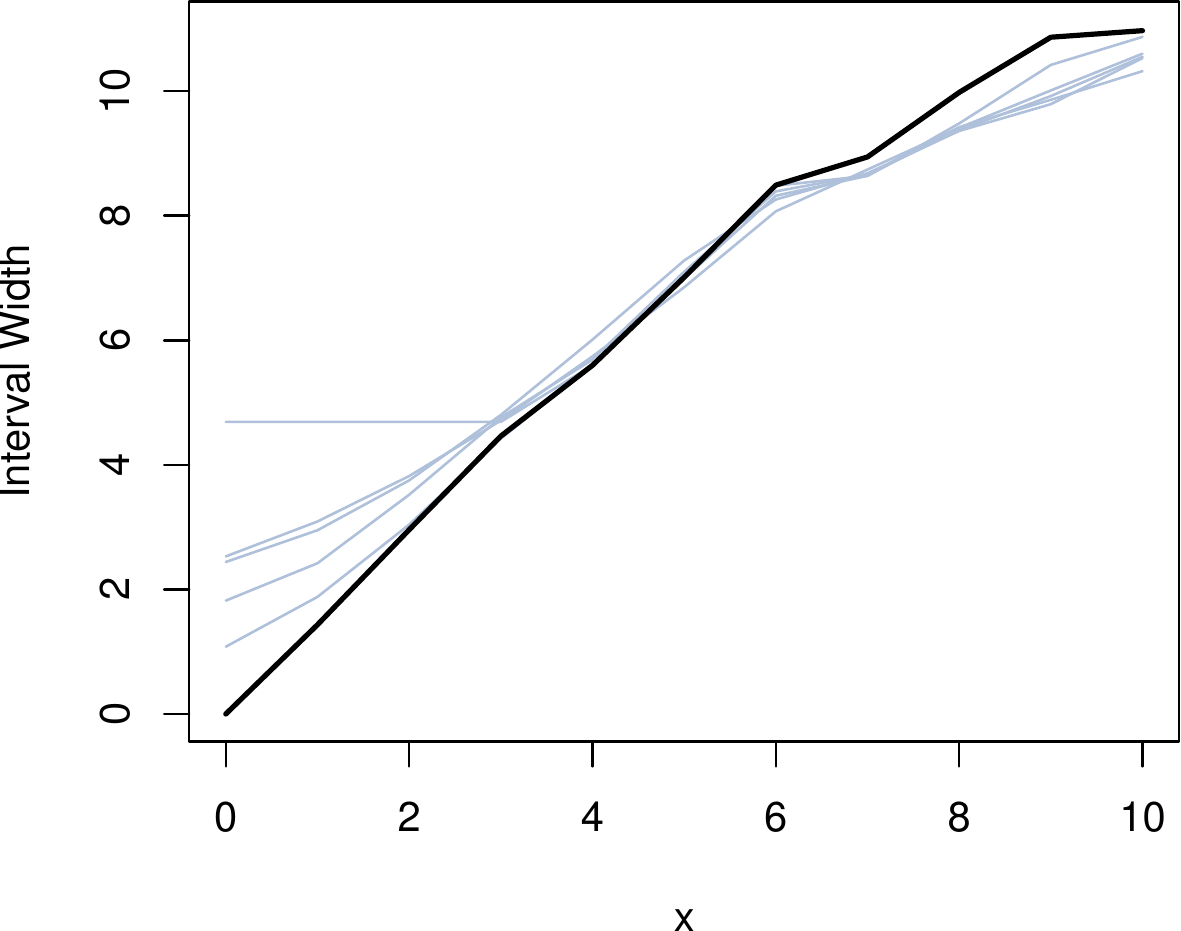}
\caption{Width of the various nominal 90\% plausibility/confidence intervals for $\lambda$, with $\beta=3$, as a function of data $x$: EB--SB (black); all others, except ELL12, (gray).}
\label{fig:intervalWidths}
\end{center}
\end{figure}

% Table of endpoints for score balanced vs. equi-tailed from IJAR with b=3

\section{Discussion}
\label{S:discuss}

Inference on a Poisson mean is an important and challenging problem, arising both classically and in modern applications.  Here we have developed a new theoretical and computational approach for optimal inference in this problem.  The main contribution is our construction of an (approximately) optimal predictive random set via a novel recursive ordering algorithm.  We also developed the EB--SB method to handle the more challenging problem of inference about a a Poisson mean when non-stochastic constraint information is available, which may be useful to high-energy physicists working on applications in this area.  Also, the techniques described herein are, for the most part, not special to the Poisson problem.  So, other challenging discrete data problems (e.g., binomial) can be handled similarly, and we expect that the corresponding optimal IM will outperform existing methods there as well.  

Numerical results focused primarily on comparing various methods in terms of frequentist performance.  But we want to reiterate once more that IMs, and the belief and plausibility functions derived from them, are more than just tools for developing frequentist procedures.  Indeed, IMs can be used to produce prior-free posterior probabilistic summaries of evidence in observed data \emph{for} and \emph{against} any assertion about the parameter of interest.  Moreover, this inferential output is meaningful both within and across experiments in the sense described in Section~\ref{S:intro}.  It is especially important that these claims hold even for singleton assertions/point null hypotheses, problems of extreme scientific importance for which existing approaches, in general, cannot give satisfactory probabilistic assessments of uncertainty.   

From a philosophical point of view, the IM framework, in general, helps tie together a number of elusive topics.  First, it identifies and corrects the inherent selection bias in Fisher's fiducial probabilities.  Roughly speaking, the fiducial probability for an assertion involves a $\prob_U$-probability calculation on a data-dependent event in $\UU$, and these probabilities tend to be too large for validity to hold.  By choosing an admissible predictive random set, the corresponding belief probability is shrunk down enough for validity to be achieved, thereby correcting the fiducial bias.  Second, by making an optimal choice of IM, the corresponding plausibility function at $A$ can be shown to equal Fisher's p-value for $H_0: \theta \in A$.  There is well-documented difficulty in interpretation of p-values, i.e., they are not bona fide probabilities for the truthfulness of $H_0$ because they require conditioning on $\theta \in A$, etc.  However, it can be shown that there exists a meaningful IM with the Fisher p-value equal to the easy-to-interpret plausibility for the truth of the claim ``$\theta \in A$''---no conditioning on the truthfulness of the claim is needed.

\section*{Acknowledgments}

This work is partially supported by the U.S.~National Science Foundation, grants DMS--1007678, DMS--1208833, and DMS--1208841.

\bibliographystyle{/Users/rgmartin/Research/TexStuff/asa}
\bibliography{/Users/rgmartin/Research/mybib}

\end{document}